\documentclass[runningheads]{llncs}

\usepackage[T1]{fontenc}
\usepackage{graphicx}
\usepackage{amsmath}
\usepackage{amssymb}
\usepackage{multicol}
\usepackage{pgffor}
\usepackage{ifthen}
\usepackage{pgfplots}
\usepackage{array,multicol}

\usepackage{tikz}
\usepackage{tikz-cd}
\usetikzlibrary{patterns,arrows, topaths, calc, positioning}
\tikzset{>=stealth}
\tikzset{>=stealth}
\tikzstyle{vertex} = [circle, minimum size = 1.5mm, inner sep = 0mm, draw={black}, fill]
\tikzstyle{hyperedge} = [rectangle, minimum width = 5mm, minimum height = 5mm, draw, inner sep = 0mm]
\tikzstyle{hyperedgewide} = [rectangle, minimum width = 8mm, minimum height = 5mm, draw, inner sep = 0mm]
\tikzstyle{HG} = [align = center]
\tikzstyle{circledge} = [circle, minimum size = 7mm, inner sep = 0mm, color=black, draw]
\tikzstyle{HG} = [align = center]
\tikzstyle{circledge} = [circle, minimum size = 7mm, inner sep = 0mm, color=black, draw]

\newcommand{\type}{\mathit{type}}
\newcommand{\Nat}{\mathbb{N}}
\newcommand{\fs}{{\mathrm{fs}}}
\newcommand{\bn}{{\mathrm{bn}}}

\newcommand{\st}{{\mathrm{st}}}

\newcommand{\BGpt}{\mathit{PsT}}
\newcommand{\BG}{\mathrm{BG}}

\newcommand{\TBG}{\mathit{TBG}}
\newcommand{\THyp}{\mathit{TH}}
\newcommand{\tv}{\mathit{tv}}
\newcommand{\Bond}{\mathcal{B}}

\newcommand{\lEnd}[1]{\multicolumn{1}{|c}{#1}}
\newcommand{\rEnd}[1]{\multicolumn{1}{c|}{#1}}
\newcommand{\lrEnd}[1]{\multicolumn{1}{|c|}{#1}}

\begin{document}

\title{Bonding Grammars}

\author{Tikhon Pshenitsyn\orcidID{0000-0003-4779-3143}}
\authorrunning{T. Pshenitsyn}

\institute{
Steklov Mathematical Institute of Russian Academy of Sciences
\\
8 Gubkina St., Moscow 119991, Russian Federation
\\
\email{tpshenitsyn at mi-ras.ru}
}
\maketitle              
\begin{abstract}
We introduce bonding grammars, a graph grammar formalism developed to model DNA computation by means of graph transformations. It is a modification of fusion grammars introduced by Kreowski, Kuske and Lye in 2017. Bonding is a graph transformation that consists of merging two hyperedges into a single larger one. We show why bonding models interaction between DNA molecules better than fusion. Then, we investigate formal properties of this formalism. Firstly, we study the relation between bonding grammars and hyperedge replacement grammars proving that each of these kinds of grammars generates a language the other one cannot generate. Secondly, we prove that bonding grammars naturally generalise regular sticker systems. Finally, we prove that the membership problem for bonding grammars is NP-complete and, moreover, that some bonding grammar generates an NP-complete set.

\keywords{graph grammar \and fusion grammar \and sticker system \and hyperedge replacement grammar \and NP-complete \and DNA computing}
\end{abstract}

\section{Introduction}\label{sec_introduction}
There is a variety of formal apporaches developed to formalise DNA computation (see the overview of those in \cite{PaunRS98}). Since a nucleotide chain can be described as a string over the alphabet $\{A,T,C,G\}$, it is natural to model DNA computation as a series of transformations of strings or of string-based structures. A prominent example of such a string-based approach is sticker systems \cite{PaunR98}. They deal with ``incomplete molecules'', structures defined as pairs of strings with certain properties, on which the sticking operation is defined. Thus, in general, most DNA-inspired formalisms belong to the area of formal grammars for strings, the area which also studies finite automata, context-free grammars etc. 

However, instead of simulating a DNA molecule by a generalised string, one could also represent it as a graph. This choice would enable one to describe the processes occurring between molecules by means of graph transformations and thus to enter the vast field of graph transformation and graph grammars (see its overview in \cite{Rozenberg97}). One could also benefit from working with graphs because, on the one hand, reasonings would become more general, while, on the other hand, definitions and proofs would become more concise. In particular, the notion of a graph is simpler and more common than the notion of an incomplete molecule from the field of sticker systems, the latter being ad-hoc to a certain degree. 

In 2017, a novel kind of graph grammars was introduced by Kreowski, Kuske and Lye called \emph{fusion grammars} \cite{KreowskiKL17}. As the authors say, it is inspired by various fusion processes occurring in nature and abstract science, and, in particular, by fusion in DNA double strands according to the Watson-Crick complementarity. The thesis \cite{Lye22} is an extensive survey of fusion grammars. In the fusion grammar approach, molecules are modeled by hypergraphs (see the formal definition of a hypergraph in Section \ref{sec_preliminaries}). Hypegraphs have labels on hyperedges; the label alphabet includes fusion labels, each fusion label $A$ being accompanied with the complementary one $\overline{A}$. Hyperedges with complementary labels can be fused, which means that their corresponding incident vertices are identified, and the hyperedges themselves are removed. See an example of fusion below.
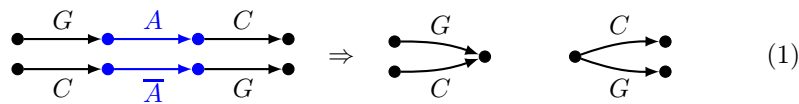
\begin{equation}\label{eq_fusion}
	\vcenter{\hbox{{\tikz[baseline=.1ex]{
					\node[vertex] (V1) at (0,0) {};
					\node[vertex, blue] (V2) at (1.2,0) {};
					\node[vertex, blue] (V3) at (2.4,0) {};
					\node[vertex] (V4) at (3.6,0) {};
					\foreach \i in {1,4}
					{
						\node[vertex] (2V\i) at ($(V\i)+(0,-0.4)$) {};
					}
					\foreach \i in {2,3}
					{
						\node[vertex,blue] (2V\i) at ($(V\i)+(0,-0.4)$) {};
					}
					\draw[-latex, thick] (V1) -- node[above] {$G$} (V2);
					\draw[-latex, thick, blue] (V2) -- node[above] {$A$} (V3);
					\draw[-latex, thick] (V3) -- node[above] {$C$} (V4);
					\draw[-latex, thick] (2V1) -- node[below] {$C$} (2V2);
					\draw[-latex, thick, blue] (2V2) -- node[below] {$\overline{A}$} (2V3);
					\draw[-latex, thick] (2V3) -- node[below] {$G$} (2V4);
	}}}}
	\quad\Rightarrow\quad
	\vcenter{\hbox{{\tikz[baseline=.1ex]{
					\node[vertex] (V1) at (0,0) {};
					\node[vertex] (V4) at (3.6,0) {};
					\node[vertex] (2V1) at (0,-0.4) {};
					\node[vertex] (2V4) at (3.6,-0.4) {};
					\node[vertex] (X1) at (1.2,-0.2) {};
					\node[vertex] (X2) at (2.4,-0.2) {};
					\draw[-latex, thick] (V1) to[bend left = 10] node[above] {$G$} (X1);
					\draw[-latex, thick] (X2) to[bend left = 10] node[above] {$C$} (V4);
					\draw[-latex, thick] (2V1) to[bend right = 10] node[below] {$C$} (X1);
					\draw[-latex, thick] (X2) to[bend right = 10] node[below] {$G$} (2V4);
	}}}}
\end{equation}

In its simplest form, a fusion grammar is defined as a tuple $Z(1),\dotsc,Z(k)$ consisting of connected hypergraphs $Z(i)$, which are informally interpreted as molecules. One can take arbitrary many exemplars of each ``molecule'' $Z(i)$ for $i=1,\dotsc,k$ and then apply fusions to the resulting collection of hypergraphs. Any connected component of the hypergraph obtained after fusions is said to be generated by the grammar. The definitions from \cite{KreowskiKL17} also involve a filtering procedure, which, however, we are not going to discuss in this paper.

Mathematically, fusion grammar is a nice graph grammar formalism with interesting relations to other graph grammars \cite{Lye22}. However, it has two drawbacks. Firstly, one would expect that, in a graph-based approach modelling DNA computing, one should represent a DNA strand as a path graph. For example, one would anticipate that the strand defined by the sequence of nucleobases $GAC$ should be represented by the graph $\vcenter{\hbox{{\tikz[baseline=.1ex]{
				\node[vertex] (V1) at (0,0) {};
				\node[vertex] (V2) at (0.8,0) {};
				\node[vertex] (V3) at (1.6,0) {};
				\node[vertex] (V4) at (2.4,0) {};
				\draw[-latex, thick] (V1) -- node[above] {$G$} (V2);
				\draw[-latex, thick] (V2) -- node[above] {$A$} (V3);
				\draw[-latex, thick] (V3) -- node[above] {$C$} (V4);
}}}}$. Then, one would expect that, in (\ref{eq_fusion}), if $\overline{A}$ stands for $T$, i.e. for the nucleobase complementary to $A$, then this fusion rule application models forming a bond between the DNA strands $GAC$ and $CTG$. However, the result of fusion presented in (\ref{eq_fusion}) does not look as two strands connected to each other. From the biological point of view, if an edge of a graph is interpreted as a nucleobase, then it is very strange that fusion results in disappearance of the nucleobases participating in fusion. Instead, one would expect that fusion of two molecules binds them together, i.e. that fusion of two connected hypergraphs always results in a larger connected hypergraph. This is not the case in (\ref{eq_fusion}). 

The second drawback of fusion grammars is that no fast algorithms are known which check whether a hypergraph is generated by a fusion grammar; the recently claimed upper complexity bound is NEXPTIME \cite{Pshenitsyn23}. Even decidability of this problem is hard to prove, which is again due to the fact that edges disappear after fusion without a trace. The question arises if the supposedly high algorithmic complexity of fusion grammars pays back; one might expect that the basic principles of DNA computing could be described by a graph formalism of less complexity.

A natural and simple idea that solves all the abovementioned problems is to change the fusion operation in the following way. Instead of making two hyperedges disappear and their attachment vertices be identified, let us combine two hyperedges into a single larger hyperedge without identifying vertices. The fusion rule application (\ref{eq_fusion}) then turns into the following one:
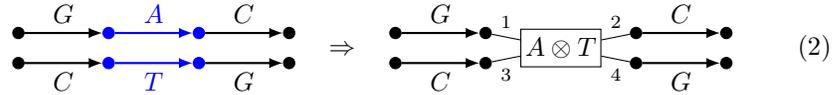
\begin{equation}\label{eq_bonding}
	\vcenter{\hbox{{\tikz[baseline=.1ex]{
					\node[vertex] (V1) at (0,0) {};
					\node[vertex, blue] (V2) at (1.2,0) {};
					\node[vertex, blue] (V3) at (2.4,0) {};
					\node[vertex] (V4) at (3.6,0) {};
					\node[vertex] (2V1) at (0,-0.4) {};
					\node[vertex, blue] (2V2) at (1.2,-0.4) {};
					\node[vertex, blue] (2V3) at (2.4,-0.4) {};
					\node[vertex] (2V4) at (3.6,-0.4) {};
					\draw[-latex, thick] (V1) -- node[above] {$G$} (V2);
					\draw[-latex, thick, blue] (V2) -- node[above] {$A$} (V3);
					\draw[-latex, thick] (V3) -- node[above] {$C$} (V4);
					\draw[-latex, thick] (2V1) -- node[below] {$C$} (2V2);
					\draw[-latex, thick, blue] (2V2) -- node[below] {$T$} (2V3);
					\draw[-latex, thick] (2V3) -- node[below] {$G$} (2V4);
	}}}}
	\quad\Rightarrow\quad
	\vcenter{\hbox{{\tikz[baseline=.1ex]{
					\node[vertex] (V1) at (0,0) {};
					\node[vertex] (V2) at ($(V1)+(1.2,0)$) {};
					\node[vertex] (V3) at ($(V2)+(2,0)$) {};
					\node[vertex] (V4) at ($(V3)+(1.2,0)$) {};
					\foreach \i in {1,2,3,4}
					{
						\node[vertex] (2V\i) at ($(V\i)+(0,-0.4)$) {};
					}
					\draw[-latex, thick] (V1) -- node[above] {$G$} (V2);
					\draw[-latex, thick] (V3) -- node[above] {$C$} (V4);
					\draw[-latex, thick] (2V1) -- node[below] {$C$} (2V2);
					\draw[-latex, thick] (2V3) -- node[below] {$G$} (2V4);
					\node[hyperedge] (E) at (2.2,-0.2) {\,$A \otimes T$\,};
					\draw[-] (E) -- node[above] {\scriptsize 1} (V2);
					\draw[-] (E) -- node[above] {\scriptsize 2} (V3);
					\draw[-] (E) -- node[below] {\scriptsize 3} (2V2);
					\draw[-] (E) -- node[below] {\scriptsize 4} (2V3);
	}}}}
\end{equation}
Here $A \otimes T$ is the label of the resulting hyperedge, which stands for the combination of $A$ and $T$. We call such an operation of combining two hyperedges into a single one \emph{bonding} because it models the situation where two nucleobases are paired with each other via a hydrogen bond thus forming a base pair. If the edges with the labels $A$ and $T$ from (\ref{eq_bonding}) are interpreted as nucleobases, then the hyperedge with the label $A \otimes T$, which is obtained after bonding, should be interpreted as a base pair. Thus bonding describes the basic interaction between DNA molecules in a natural way.

On the basis of bonding, in Section \ref{sec_bonding}, we define the notion of \emph{bonding grammar} analogous to the notion of fusion grammar. In Section \ref{ssec_properties}, we study properties of this new kind of graph grammars, in particular, its relation to hypergraph context-free grammars (also known as hyperedge replacement grammars \cite{Rozenberg97}). We show that either of the two formalisms is able to generate a language the other one cannot generate. In Section \ref{ssec_bonding_sticker}, we show that bonding grammars naturally extend regular sticker systems, and thus the former can be considered as a hypergraph counterpart of the latter. In Section \ref{ssec_bonding_membership}, we prove that the membership problem for bonding grammars is NP-complete, which is better than the known NEXPTIME algorithm for fusion grammars. This shows that bonding grammars have reasonable level of complexity, many problems from the graph theory being NP-complete. To prove NP-hardness we use the problem of partitioning a graph into triangles for graphs of bounded degree \cite{vanRooijvKNB12}.

\section{Preliminaries}\label{sec_preliminaries}
By $[k]$ we denote the set $\{1,\dotsc,k\}$. In this work, elements of $A^k$ are treated either as strings (words), as tuples or as multisets with a fixed linear order. If $w \in A^k$, then its $i$-th component is denoted by $w(i)$ (for $i \in [k]$). If $w \in A^k$, then $\vert w \vert = k$ is the length of $w$ as a string. As usually, $A^\ast$ is the union of $A^k$ for $k \in \Nat$, and $A^0 = \{\lambda\}$ consists of the empty word. The notation $a \in w$ for $w \in A^k$ stands for $a \in \{w(1),\dotsc,w(k)\}$.

\subsubsection{Hypergraphs}\label{ssec_preliminaries_hypergraph}
Our definition of a hypergraph is similar to that from \cite{KreowskiKL17} where hypergraphs are directed and have labels on hyperedges. This way of defining hypergraph is common in the literature on hyperedge replacement grammars \cite{Engelfriet97}.

A \emph{typed set} $M$ is the set $M$ along with a fixed type function $\type:M \to \Nat$. 
\begin{definition}
	Given a typed set $\Sigma$ called \emph{the label alphabet}, a \emph{hypergraph $H$ over $\Sigma$} is a tuple $H = (V, E, att, lab)$ where $V$ is a finite set of \emph{vertices}, $E$ is a finite set of \emph{hyperedges}, $att: E \to V^\ast$ is an \emph{attachment} (a function that assigns a string consisting of vertices to each hyperedge; these vertices are called \emph{attachment vertices}), and $lab: E \to \Sigma$ is a \emph{labeling} such that $\type(lab(e))=\vert att(e)\vert$ for any $e\in E_H$.
	The components of $H$ are denoted by $V_H$, $E_H$, $att_H$ and $lab_H$. 
	The set of hypergraphs over $\Sigma$ is denoted by $\mathcal{H}(\Sigma)$. 
\end{definition}
In drawings of hypergraphs, small circles represent vertices, labeled rectangles represent hyperedges with their labels, and attachment is represented by numbered lines. If a hyperedge has exactly two attachment vertices, then it is depicted by a labeled arrow that goes from the first attachment vertex to the second one.
\begin{example}\label{example_hypergraph}
	The hypergraph $(V,E,att,lab)$ where $V = \{v_1,v_2,v_3\}$, $E = \{e_1,e_2,e_3\}$, $att(e_1) = v_1$, $att(e_2) = v_1v_2v_3$, $att(e_3) = v_2v_2$, $lab(e_1) = X$, $lab(e_2) = Y$, $lab(e_3) = Z$ looks as follows:
	$\vcenter{\hbox{{\tikz[baseline=.1ex]{
					\node[hyperedge] (E1) at (0,0) {$X$};
					\node[hyperedge] (E2) at (1.3,0) {$Y$};
					\node[vertex] (V1) at (0.65,0) {};
					\node[vertex] (V2) at (1.95,0) {};
					\node[vertex] (V3) at (1.3,0.65) {};
					\draw[latex-, thick] (V2) to[out=-20,in=20,looseness=30] node[right] {$Z$} (V2);
					\draw[-] (E2) -- node[above] {\scriptsize 1} (V1);
					\draw[-] (E2) -- node[above] {\scriptsize 2} (V2);
					\draw[-] (E2) -- node[right] {\scriptsize 3} (V3);
					\draw[-] (E1) -- (V1);
	}}}}$. Here the label $X$ has type 1, the label $Y$ has type 3, and the label $Z$ has type 2.
\end{example}

The function $\type: E_H \to \Nat$ returns the number of attachment vertices of a hyperedge in the hypergraph $H$: $\type(e) = \vert att_H(e) \vert$. 

Let $H$ be a hypergraph. A sequence $(i_1,e_1,o_1),\dotsc,(i_n,e_n,o_n) \in (\Nat \times E_H \times \Nat)^\ast$ is a \emph{path between $v \in V_H$ and $v^\prime \in V_{H}$} if $att_H(e_1)(i_1)=v$, $att_H(e_n)(o_n) = v^\prime$, and $att_H(e_k)(o_k) = att_H(e_{k+1})(i_{k+1})$ for $k=1,\dotsc,n-1$. A hypergraph is \emph{connected} if there is a path between any two vertices. 

The notion of an isomorphism of hypergraphs is standard. As is usually done in the literature on graph grammars (see discussion in \cite{Rozenberg97}), we often identify isomorphic hypergraphs because we are interested in their structure rather than in what objects their vertices and hyperedges are. So, the notation $G = H$ usually means that $G$ and $H$ are isomorphic.

Below we define several operations on hypergraphs. 
\begin{enumerate}
	
	\item The \emph{disjoint union} $H_1+H_2$ of hypergraphs $H_1$, $H_2$ is the hypergraph $( V_{H_1}\sqcup V_{H_2}, E_{H_1}\sqcup E_{H_2}, att, lab )$ such that $att|_{H_i} = att_{H_i}$, $lab|_{H_i} = lab_{H_i}$ ($i=1,2$). That is, we just put the hypergraphs $H_1$ and $H_2$ together without affecting their structures. 
	Let $k \cdot H$ be the disjoint union of $k$ copies of $H$. Let $0\cdot H$ be the hypergraph with no vertices and no hyperedges. 
	
	\item Given a vector $m = (m(1),\dotsc,m(k)) \in \Nat^k$ and a tuple of hypergraphs $H = (H(1),\dotsc,H(k))$, let $m \cdot H$ denote $m(1) \cdot H(1)+\dotsc+m(k)\cdot H(k)$.
	
	\item Let $H$ be a hypergraph and let $E \subseteq E_{H}$. Then $H - E$ is the hypergraph
	\\ $(V_{H},E_{H} \setminus E, att_H\restriction_{E_{H} \setminus E},lab_H \restriction_{E_{H} \setminus E})$; i.e. we simply remove hyperedges that belong to $E$ from $H$.
	
	\item Let $H$ be a hypergraph and let $\equiv$ be an equivalence relation on $V_H$. Then $H/{\equiv}$ is the hypergraph with the set $\{[v]_\equiv \mid v \in V_H\}$ of vertices, the set $E_H$ of hyperedges, the same labeling $lab_H$ and with the new attachment defined as follows: $att_{H/\equiv}(e) = [att_{H}(e)(1)]_\equiv [att_{H}(e)(2)]_\equiv \dotsc [att_{H}(e)(\type(e))]_\equiv$.
\end{enumerate}

Sets of hypergraphs are called \emph{hypergraph languages}. The ``canonical'' formalism for generating hypergraph languages is \emph{hyperedge replacement grammar}, also known as hypergraph context-free grammar. It naturally generalises context-free grammars for strings and, in particular, enjoys many similar properties, such as the pumping lemma, the Parikh theorem etc. We shall compare generative capacity of bonding grammars with that of hyperedge replacement grammars. However, we are not going to provide the formal definition of the latter because we do not need it for reasonings. We refer the reader to \cite{Engelfriet97,Rozenberg97} for the definitions.

\subsubsection{Fusion Grammar}\label{ssec_preliminaries_FG}

Fusion grammar was introduced in \cite{KreowskiKL17}. After that, its definition has undergone slight modifications, its most recent version is presented in \cite{Lye22}. There, the notion of fusion grammars is based on the fusion operation along with additional procedures such as filtering the hypergraphs using markers and removing hyperedges labeled by connector labels. In this paper, we are not interested in the filtering procedure as well as in connector labels, the latter being biologically ungrounded and being used in \cite{Lye22} for technical purposes. The formalism we are going to introduce is called \emph{fusion grammar without markers} in \cite{KreowskiKL17}; we, however, call it simply \emph{fusion grammar} for the sake of brevity.

Let $T$ and $F$ be two disjoint finite typed alphabets called the \emph{terminal alphabet} and the \emph{fusion alphabet} resp. Let $\overline{F} = \{\overline{A} \mid A \in F\}$ be the alphabet of complementary fusion labels.
\begin{definition}
	A \emph{fusion grammar} is a triple $FG = (Z,F,T)$ where \\ $Z = (Z(1),\dotsc,Z(k))$ is a tuple of connected hypergraphs $Z(i) \in \mathcal{H}(T \cup F \cup \overline{F})$.
\end{definition}
\begin{definition}
	Let $H$ be a hypergraph with two hyperedges $e,\overline{e} \in E_H$ such that $lab_H(e) = A \in F$ and $lab_H\left(\overline{e}\right) = \overline{A} \in \overline{F}$. Let $X = H - \{e,\overline{e}\}$ be obtained from $H$ by removing $e$ and $\overline{e}$. Let $H^\prime = X/{\equiv}$ where $\equiv$ is the smallest equivalence relation on $V_H$ such that $att_H(e)(i) \equiv att_H(\overline{e})(i)$ for all $i =1, \dotsc, \type(A)$. In other words, $H^\prime$ is obtained from $X$ by identifying corresponding attachment vertices of $e$ and $\overline{e}$. Then we say that $H^\prime$ is obtained from $H$ by \emph{fusion} and denote this as $H \Rightarrow_\fs H^\prime$.
\end{definition}
\begin{definition}\label{def_FG_generates_hypergraph}
	A fusion grammar $FG = (Z,F,T)$ such that $\vert Z \vert = k$ \emph{generates a hypergraph $H$} if there exists $m \in \Nat^k$ such that $m \cdot Z \Rightarrow_\fs^\ast G$, and $H$ is a connected component of $G$.
\end{definition}

\begin{example}
	Let $(Z,F,T)$ be the fusion grammar where $Z = (Z(1),Z(2))$ consists of the hypergraphs $Z(1) = \vcenter{\hbox{{\tikz[baseline=.1ex]{
					\node[vertex] (V1) at (0,0) {};
					\node[vertex] (V2) at (0.8,0) {};
					\node[vertex] (V3) at (1.6,0) {};
					\node[vertex] (V4) at (2.4,0) {};
					\draw[-latex, thick] (V1) -- node[above] {$\overline{C}$} (V2);
					\draw[-latex, thick] (V2) -- node[above] {$A$} (V3);
					\draw[-latex, thick] (V3) -- node[above] {$C$} (V4);
	}}}}$ 
	and 
	$Z(2) = \vcenter{\hbox{{\tikz[baseline=.1ex]{
	\node[vertex] (V1) at (0,0) {};
	\node[vertex] (V2) at (0.8,0) {};
	\node[vertex] (V3) at (1.6,0) {};
	\node[vertex] (V4) at (2.4,0) {};
	\draw[-latex, thick] (V1) -- node[above] {$C$} (V2);
	\draw[-latex, thick] (V2) -- node[above] {$\overline{A}$} (V3);
	\draw[-latex, thick] (V3) -- node[above] {$\overline{C}$} (V4);
	}}}}$. Then this fusion grammar generates the hypergraph 
	$
	H_0 = \vcenter{\hbox{{\tikz[baseline=.1ex]{
		\node[vertex] (V1) at (0,0) {};
		\node[vertex] (V2) at (0.8,0) {};
		\node[vertex] (V3) at (1.6,0) {};
		\draw[-latex, thick] (V1) -- node[above] {$\overline{C}$} (V2);
		\draw[latex-, thick] (V2) -- node[above] {$C$} (V3);
	}}}}$, as (\ref{eq_fusion}) shows. Indeed, it is shown there that $Z(1) + Z(2) \Rightarrow_\fs H$ where $H$ contains $H_0$ as a connected component (assuming that $G = \overline{C}$ and $T = \overline{A}$).
\end{example}

\subsubsection{Sticker Systems}
An overview of sticker systems can be found in \cite{PaunRS98,PaunR98}. We are going to use the definitions from \cite{KutribW21}. In the original definition, sticker systems generate double strands (pairs of strings) such that symbols in the upper strand are in the complementarity relation with those in the lower strand. However, for the sake of simplifying presentation, following \cite{KutribW21}, we consider only sticker systems with the identity complementarity relation. 

Let $\Sigma$ be a finite alphabet. The set $D^=_\Sigma$ consists of complete double strands of the form
$\begin{bmatrix}
	w \\ w \\
\end{bmatrix}$ for $w \in \Sigma^\ast$, $w \ne \lambda$. The set $E_\Sigma$ of sticky ends consists of pairs of strings of the form 
$\begin{pmatrix}
	v \\ \lambda \\
\end{pmatrix}$
and
$\begin{pmatrix}
	\lambda \\ v \\ 
\end{pmatrix}$
where $v \in \Sigma^\ast$; if $v = \lambda$, then the corresponding sticky end is denoted by $\lambda$. Let $LR_\Sigma = E_\Sigma \cdot D^=_\Sigma \cdot E_\Sigma$ and let $R_\Sigma = D^=_\Sigma \cdot E_\Sigma$; here the product stands for juxtaposition (and $\lambda$ is the unit of this product). Finally, the set $D_\Sigma$ of \emph{dominoes} (called \emph{incomplete molecules} in \cite{PaunRS98}) equals $E_\Sigma \cup LR_\Sigma$. 
Below, we present several examples of dominoes along with another way of their representation, which we call \emph{domino representation}.

$$
\begin{pmatrix}
	ab \\ \lambda \\
\end{pmatrix}
\begin{bmatrix}
	cd \\ cd \\
\end{bmatrix}
\begin{pmatrix}
	\lambda \\ e \\
\end{pmatrix}
=
\vcenter{\hbox{
		\begin{tabular}{ccccc}
			\cline{1-4}
			\lEnd{$a$} & $b$ & $c$ & \rEnd{$d$} &     \\\cline{1-2}\cline{5-5}
			&     & \lEnd{$c$} & $d$ & \rEnd{$e$} \\
			\cline{3-5}
		\end{tabular}
}}
;\qquad
\begin{bmatrix}
	aa \\ aa \\
\end{bmatrix}
\begin{pmatrix}
	e \\ \lambda \\ 
\end{pmatrix}
=
\vcenter{\hbox{
		\begin{tabular}{|ccc}
			\cline{1-3}
			$a$ & $a$ & \rEnd{$e$} \\\cline{3-3}
			$a$ & \rEnd{$a$} &  \\
			\cline{1-2}
		\end{tabular}
}}
;\qquad
\begin{pmatrix}
	\lambda \\ abc \\
\end{pmatrix}
=
\vcenter{\hbox{
		\begin{tabular}{ccc}
			&  & \\\cline{1-3}
			\lEnd{$a$} & $b$ & \rEnd{$c$} \\
			\cline{1-3}
		\end{tabular}
}}
.
$$

Given $x \in LR_\Sigma$ and $y \in D_\Sigma$, one can define the sticking operation $x \cdot y$ and $y \cdot x$. Informally, sticking $x \cdot y$ of $x$ and $y$ consists of placing $y$ to the right from $x$ and then assembling them into a single domino if the right sticky end of $x$ matches the left sticky end of $y$. For example:
$$
\begin{pmatrix}
	ab \\ \lambda \\
\end{pmatrix}
\begin{bmatrix}
	cd \\ cd \\
\end{bmatrix}
\begin{pmatrix}
	\lambda \\ e \\
\end{pmatrix}
\cdot
\begin{pmatrix}
	e \\ \lambda \\ 
\end{pmatrix}
\begin{bmatrix}
	aab \\ aab \\
\end{bmatrix}
=
\begin{pmatrix}
	ab \\ \lambda \\
\end{pmatrix}
\begin{bmatrix}
	cdeaab \\ cdeaab \\
\end{bmatrix}
=
\vcenter{\hbox{
		\begin{tabular}{cccccccc|}
			\cline{1-8}
			\lEnd{$a$} & $b$ & $c$ & $d$ &  $e$ & $a$ & $a$ & $b$   \\\cline{1-2}
			&     & \lEnd{$c$} & $d$ & $e$ & $a$ & $a$ & $b$ \\
			\cline{3-8}
		\end{tabular}
}}
$$
The formal definition of sticking can be found in \cite{PaunRS98} (it is too large to provide it here). A \emph{sticking rule} is a pair of dominoes $r = (d_1,d_2) \in D_\Sigma^2$ (at least one domino in $r$ must be non-empty). An application of $r$ to $u \in LR_\Sigma$ consists of sticking $d_1$ to the left from $u$ and $d_2$ to the right from $u$ resulting in $u^\prime = d_1 \cdot u \cdot d_2$. This is denoted by $u \Rightarrow_\st u^\prime$. A \emph{sticker system} is a triple $SS = (\Sigma, A, D)$ where $A \subseteq LR_\Sigma$ is a finite set of axioms and $D$ is a finite set of sticking rules. We say that a sticker system \emph{generates} $d \in LR_\Sigma$ if $a \Rightarrow_\st^\ast d$ for some $a \in A$. A sticker system $(\Sigma,A,D)$ is \emph{regular} if each rule in $D$ is of the form $(\lambda,d)$ and if $A \subseteq R_\Sigma$. Note that, if a domino $d$ is generated by a regular sticker system, then $d \in R_\Sigma$.

\section{Bonding Grammar}\label{sec_bonding}

Our goal is to develop a graph grammar formalism to adequately model DNA computation. To make it biologically grounded, we adhere to the following principle: a hyperedge of a hypergraph must either represent a nucleotide or a bond between nucleotides. Fusion grammars do not meet this requirement because, if two hyperedges $e_1$ and $e_2$ are fused, then they disappear, which is incorrect from the biological point of view (how can two compounds vanish after being connected?). During annealing of two single DNA strands, hydrogen bonds arise, and pairs of bonded nucleobases form new units called \emph{base pairs}. This leads us to the following natural idea: if $e_1$ and $e_2$ are two hyperedges representing two complementary nucleobases, then the process of fusing them should be modeled by a graph transformation where these hyperedges merge into a single hyperedge. This idea is formalised in the definitions of \emph{bonding} and of \emph{bonding grammar}.
\begin{definition}
	A \emph{bonding grammar} is a tuple $(Z,N,T,\otimes)$ where
	\begin{enumerate}
		\item $N$ and $T$ are disjoint finite typed sets;
		\item $Z = (Z(1),\dotsc,Z(k))$ is a tuple of connected hypergraphs $Z(i) \in \mathcal{H}(N \cup T)$;
		\item $\otimes:N\times N \to T$ is a partial injective function called the \emph{bond function} such that, if $A_1 \otimes A_2$ is defined, then $\type(A_1 \otimes A_2) = \type(A_1) + \type(A_2)$.
	\end{enumerate}
\end{definition}
The bond function takes two labels corresponding to nucleobases and returns a label that denotes the resulting base pair. E.g. if $C,G \in N$, then $C \otimes G$ denotes the base pair of cytosine and guanine. We consider the label $C \otimes G$ to be terminal. The domain of the bond function is interpreted as the complementarity relation on $N$, which describes between which pairs from $N\times N$ bonds can arise.

It is natural to assume that the bond function is injective because, if a nucleobase $A_1$ is paired with a nucleobase $A_2$ and $B_1$ is paired with $B_2$ for $(A_1,A_2) \ne (B_1,B_2)$, then the resulting base pairs are also different, so $A_1 \otimes A_2 \ne B_1 \otimes B_2$. Of course, one could define bonding grammars without the requirement on $\otimes$ being injective. However, this restriction simplifies technical arguments; moreover, bond functions are indeed injective in all the examples of bonding grammars we are interested in. Thus we decided to include injectivity of $\otimes$ in the definition.
\begin{definition}
	Let $H$ be a hypergraph with two distinct hyperedges $e_1,e_2 \in E_H$ such that $lab_H(e_i) = A_i \in N$ for $i=1,2$. Assume that $A_1 \otimes A_2$ is defined. Let us define the hypergraph $H^\prime$ as follows.
	\begin{itemize}
		\item $V_{H^\prime} = V_H$; \qquad $E_{H^\prime} = (E_H \setminus \{e_1,e_2\}) \cup \{e^\prime\}$ where $e^\prime$ is a new hyperedge;
		\item $att_{H^\prime}(e) = att_H(e)$ and $lab_{H^\prime}(e) = lab_H(e)$ for $e \ne e^\prime$;
		\item $att_{H^\prime}(e^\prime) = att_H(e_1)att_H(e_2)$ and $lab_{H^\prime}(e^\prime) = A_1 \otimes A_2$.
	\end{itemize}
	We say that $H^\prime$ is obtained from $H$ by \emph{bonding} and denote this by $H \Rightarrow_{BG} H^\prime$ or by $H \Rightarrow_\bn H^\prime$ (if the grammar is clear from the context).
\end{definition}

Informally, the bonding operation is joining two hyperedges of a hypergraph labeled by $A_1$ and $A_2$ into a single hyperedge labeled by $A_1 \otimes A_2$. For example, if $\type(A_1) = \type(A_2) = 2$, then bonding looks as follows:
$$
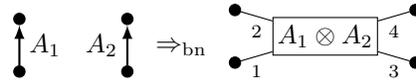
\vcenter{\hbox{{\tikz[baseline=.1ex]{
				\node[vertex] (V1) at (0,0) {};
				\node[vertex] (V2) at (0,0.7) {};
				\draw[-latex, thick] (V1) -- node[right] {$A_1$} (V2);
}}}}
\;
\vcenter{\hbox{{\tikz[baseline=.1ex]{
				\node[vertex] (V1) at (0,0) {};
				\node[vertex] (V2) at (0,0.7) {};
				\draw[-latex, thick] (V1) -- node[left] {$A_2$} (V2);
}}}}
\;\;\Rightarrow_\bn\;\;
\vcenter{\hbox{{\tikz[baseline=.1ex]{
				\node[vertex] (V11) at (0,0) {};
				\node[vertex] (V21) at (0,0.7) {};
				\node[vertex] (V12) at (2.4,0) {};
				\node[vertex] (V22) at (2.4,0.7) {};
				\node[hyperedge] (E) at (1.2,0.35) {\,$A_1 \otimes A_2$\,};
				\draw[-] (E) -- node[below] {\scriptsize 1} (V11);
				\draw[-] (E) -- node[below] {\scriptsize 3} (V12);
				\draw[-] (E) -- node[below] {\scriptsize 2} (V21);
				\draw[-] (E) -- node[below] {\scriptsize 4} (V22);
}}}}
$$
\begin{example}\label{example_bonding_DNA}
	Using bonding, the process of pairing two single DNA strands, say, $GAC$ and $CTG$, is represented as follows (for $C \otimes G$ and $A \otimes T$ being defined):
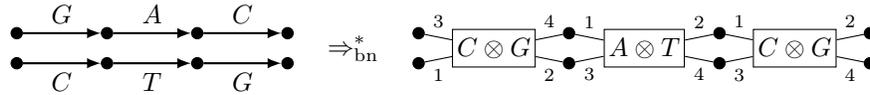
\begin{equation*}
	\vcenter{\hbox{{\tikz[baseline=.1ex]{
					\node[vertex] (V1) at (0,0) {};
					\node[vertex] (V2) at (1.2,0) {};
					\node[vertex] (V3) at (2.4,0) {};
					\node[vertex] (V4) at (3.6,0) {};
					\node[vertex] (2V1) at (0,-0.4) {};
					\node[vertex] (2V2) at (1.2,-0.4) {};
					\node[vertex] (2V3) at (2.4,-0.4) {};
					\node[vertex] (2V4) at (3.6,-0.4) {};
					\draw[-latex, thick] (V1) -- node[above] {$G$} (V2);
					\draw[-latex, thick] (V2) -- node[above] {$A$} (V3);
					\draw[-latex, thick] (V3) -- node[above] {$C$} (V4);
					\draw[-latex, thick] (2V1) -- node[below] {$C$} (2V2);
					\draw[-latex, thick] (2V2) -- node[below] {$T$} (2V3);
					\draw[-latex, thick] (2V3) -- node[below] {$G$} (2V4);
	}}}}
	\quad\Rightarrow_\bn^\ast\quad 
	\vcenter{\hbox{{\tikz[baseline=.1ex]{
					\node[vertex] (V1) at (0,0) {};
					\node[vertex] (V2) at ($(V1)+(2,0)$) {};
					\node[vertex] (V3) at ($(V2)+(2,0)$) {};
					\node[vertex] (V4) at ($(V3)+(2,0)$) {};
					\foreach \i in {1,2,3,4}
					{
						\node[vertex] (2V\i) at ($(V\i)+(0,-0.4)$) {};
					}
					\node[hyperedge] (E) at (3,-0.2) {\,$A \otimes T$\,};
					\draw[-] (E) -- node[above] {\scriptsize 1} (V2);
					\draw[-] (E) -- node[above] {\scriptsize 2} (V3);
					\draw[-] (E) -- node[below] {\scriptsize 3} (2V2);
					\draw[-] (E) -- node[below] {\scriptsize 4} (2V3);
					\node[hyperedge] (E2) at (1,-0.2) {\,$C \otimes G$\,};
					\draw[-] (E2) -- node[above] {\scriptsize 3} (V1);
					\draw[-] (E2) -- node[above] {\scriptsize 4} (V2);
					\draw[-] (E2) -- node[below] {\scriptsize 1} (2V1);
					\draw[-] (E2) -- node[below] {\scriptsize 2} (2V2);
					\node[hyperedge] (E3) at (5,-0.2) {\,$C \otimes G$\,};
					\draw[-] (E3) -- node[above] {\scriptsize 1} (V3);
					\draw[-] (E3) -- node[above] {\scriptsize 2} (V4);
					\draw[-] (E3) -- node[below] {\scriptsize 3} (2V3);
					\draw[-] (E3) -- node[below] {\scriptsize 4} (2V4);
	}}}}
\end{equation*}
\end{example}

\begin{definition}\label{def_BG_generates_hypergraph}
	A bonding grammar $BG = (Z,N,T,\otimes)$ such that $\vert Z \vert = k$ \emph{generates a hypergraph $H \in \mathcal{H}(N \cup T)$} if there exists $m \in \Nat^k$ such that $m \cdot Z \Rightarrow_{BG}^\ast H$.
	\\
	The \emph{language $L(BG)$} consists of hypergraphs from $\mathcal{H}(T)$ generated by $BG$.
\end{definition}
\begin{example}\label{example_pseudotori}
	The following example is inspired by \cite[Example 2]{KreowskiKL17}. Consider the bonding grammar $\BGpt = (Z,N,T,\otimes)$ with $N = \{A,B,C,D\}$ (all the labels are of type $1$) and $T = \{a,b\}$ (all the labels are of type $2$). Let $\otimes$ be defined as follows: $A \otimes C = a$, $B \otimes D = b$. Finally, let $Z$ consist of the hypergraph
	$Z(1) =
	\vcenter
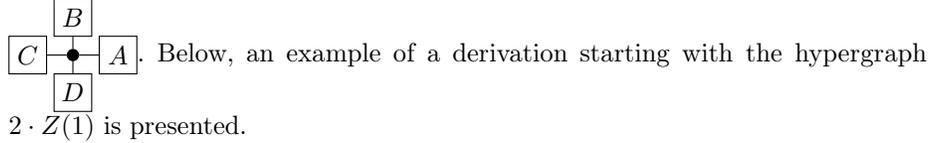
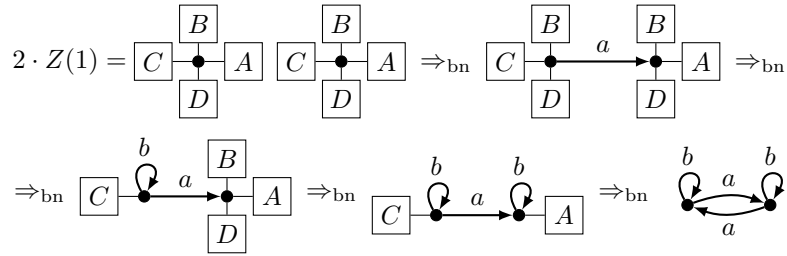
{\hbox{{\tikz[baseline=.1ex]{
					\node[vertex] (VO) at (0,0) {};
					\node[hyperedge] (EW) at (-0.6,0) {$C$};
					\node[hyperedge] (EE) at (0.6,0) {$A$};
					\node[hyperedge] (EN) at (0,0.5) {$B$};
					\node[hyperedge] (ES) at (0,-0.5) {$D$};
					\draw[-] (EW) -- (VO);
					\draw[-] (EE) -- (VO);
					\draw[-] (EN) -- (VO);
					\draw[-] (ES) -- (VO);
	}}}}$. Below, an example of a derivation starting with the hypergraph $2 \cdot Z(1)$ is presented.
	\begin{eqnarray*}
		2 \cdot Z(1) = \vcenter{\hbox{{\tikz[baseline=.1ex]{
						\node[vertex] (VO) at (0,0) {};
						\node[hyperedge] (EW) at (-0.6,0) {$C$};
						\node[hyperedge] (EE) at (0.6,0) {$A$};
						\node[hyperedge] (EN) at (0,0.5) {$B$};
						\node[hyperedge] (ES) at (0,-0.5) {$D$};
						\draw[-] (EW) -- (VO);
						\draw[-] (EE) -- (VO);
						\draw[-] (EN) -- (VO);
						\draw[-] (ES) -- (VO);
						\node[vertex] (2VO) at (1.9+0,0) {};
						\node[hyperedge] (2EW) at (1.9+-0.6,0) {$C$};
						\node[hyperedge] (2EE) at (1.9+0.6,0) {$A$};
						\node[hyperedge] (2EN) at (1.9+0,0.5) {$B$};
						\node[hyperedge] (2ES) at (1.9+0,-0.5) {$D$};
						\draw[-] (2EW) -- (2VO);
						\draw[-] (2EE) -- (2VO);
						\draw[-] (2EN) -- (2VO);
						\draw[-] (2ES) -- (2VO);
		}}}}
		\;\Rightarrow_\bn\;
		\vcenter{\hbox{{\tikz[baseline=.1ex]{
						\node[vertex] (VO) at (0,0) {};
						\node[hyperedge] (EW) at (-0.6,0) {$C$};
						\node[hyperedge] (EN) at (0,0.5) {$B$};
						\node[hyperedge] (ES) at (0,-0.5) {$D$};
						\draw[-] (EW) -- (VO);
						\draw[-] (EN) -- (VO);
						\draw[-] (ES) -- (VO);
						\node[vertex] (2VO) at (1.4+0,0) {};
						\node[hyperedge] (2EE) at (1.4+0.6,0) {$A$};
						\node[hyperedge] (2EN) at (1.4+0,0.5) {$B$};
						\node[hyperedge] (2ES) at (1.4+0,-0.5) {$D$};
						\draw[-] (2EE) -- (2VO);
						\draw[-] (2EN) -- (2VO);
						\draw[-] (2ES) -- (2VO);
						\draw[-latex, thick] (VO) -- node[above] {$a$} (2VO);
		}}}}
		\;\Rightarrow_\bn
		\\
		\Rightarrow_\bn\;
		\vcenter{\hbox{{\tikz[baseline=.1ex]{
						\node[vertex] (VO) at (0,0) {};
						\node[hyperedge] (EW) at (-0.6,0) {$C$};
						\draw[-] (EW) -- (VO);
						\node[vertex] (2VO) at (1.1+0,0) {};
						\node[hyperedge] (2EE) at (1.1+0.6,0) {$A$};
						\node[hyperedge] (2EN) at (1.1+0,0.5) {$B$};
						\node[hyperedge] (2ES) at (1.1+0,-0.5) {$D$};
						\draw[-] (2EE) -- (2VO);
						\draw[-] (2EN) -- (2VO);
						\draw[-] (2ES) -- (2VO);
						\draw[-latex, thick] (VO) -- node[above] {$a$} (2VO);
						\draw[latex-, thick] (VO) to[out=60,in=120,looseness=16] node[above] {$b$} (VO);
		}}}}
		\;\Rightarrow_\bn\;
		\vcenter{\hbox{{\tikz[baseline=.1ex]{
						\node[vertex] (VO) at (0,0) {};
						\node[hyperedge] (EW) at (-0.6,0) {$C$};
						\draw[-] (EW) -- (VO);
						\node[vertex] (2VO) at (1.1+0,0) {};
						\node[hyperedge] (2EE) at (1.1+0.6,0) {$A$};
						\draw[-] (2EE) -- (2VO);
						\draw[-latex, thick] (VO) -- node[above] {$a$} (2VO);
						\draw[latex-, thick] (VO) to[out=60,in=120,looseness=16] node[above] {$b$} (VO);
						\draw[latex-, thick] (2VO) to[out=60,in=120,looseness=16] node[above] {$b$} (2VO);
		}}}}
		\;\Rightarrow_\bn\;
		\vcenter{\hbox{{\tikz[baseline=.1ex]{
						\node[vertex] (VO) at (0,0) {};
						\node[vertex] (2VO) at (1.1+0,0) {};
						\draw[-latex, thick] (VO) to[bend left = 20] node[above] {$a$} (2VO);
						\draw[-latex, thick] (2VO) to[bend left = 20] node[below] {$a$} (VO);
						\draw[latex-, thick] (VO) to[out=60,in=120,looseness=16] node[above] {$b$} (VO);
						\draw[latex-, thick] (2VO) to[out=60,in=120,looseness=16] node[above] {$b$} (2VO);
		}}}}
	\end{eqnarray*}
	More generally, the bonding grammar $\BGpt$ is able to generate grids with $a$-labeled and $b$-labeled hyperedges and with borders marked by hyperedges of type $1$, as shown on the below left figure:
	\begin{center}
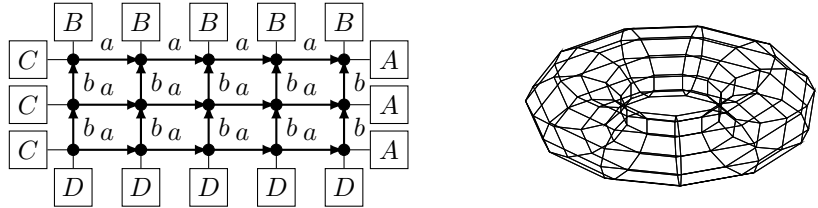
		\begin{minipage}{0.55\textwidth}
			\centering
			\begin{tikzpicture}
				\foreach \i in {0,1,2,3,4}
				{
					\foreach \j in {0,1,2}
					{
						\node[vertex] at (\i*0.9,\j*0.6) {};
						\ifthenelse{\i < 4}
						{
							\draw[-latex, thick] (\i*0.9,\j*0.6) -- node[above] {$a$} (\i*0.9+0.9,\j*0.6);
						}{}
						\ifthenelse{\j < 2}
						{
							\draw[-latex, thick] (\i*0.9,\j*0.6) -- node[right] {$b$} (\i*0.9,{(\j+1)*0.6});
						}{}
						\ifthenelse{\i = 0}
						{
							\node[hyperedge] (W\i) at ({\i*0.9-0.6},\j*0.6) {$C$};
							\draw[-] (W\i) -- (\i*0.9,\j*0.6);
						}
						{}
						\ifthenelse{\i = 4}
						{
							\node[hyperedge] (E\i) at ({\i*0.9+0.6},\j*0.6) {$A$};
							\draw[-] (E\i) -- (\i*0.9,\j*0.6);
						}
						{}
					}
					\node[hyperedge] (S\i) at (\i*0.9,-0.5) {$D$};
					\draw[-] (S\i) -- (\i*0.9,0);
					\node[hyperedge] (N\i) at (\i*0.9,0.6*2+0.5) {$B$};
					\draw[-] (N\i) -- (\i*0.9,0.6*2);
				}
			\end{tikzpicture}
		\end{minipage}
		\begin{minipage}{0.4\textwidth}
			\begin{tikzpicture}
				\begin{axis}
					[
					view={45}{65},
					hide axis,
					width=7cm,
					height=5cm
					]
					\addplot3[
					mesh,
					color=black,
					samples=12,
					domain=0:2*pi,y domain=0:2*pi,
					z buffer=sort
					]
					({(1+0.5*cos(deg(x)))*cos(deg(y+pi/2))}, 
					{(1+0.5*cos(deg(x)))*sin(deg(y+pi/2))}, 
					{sin(deg(x))});
				\end{axis}
			\end{tikzpicture}
		\end{minipage}
	\end{center}
	Clearly, if one also applies bonding to the remaining nonterminal hyperedges in this grid, they can obtain a toroidal graph. The shape of a toroidal graph is illustrated by the above right figure (the labels and directions of hyperedges are omitted). Therefore, the language $L(\BGpt)$ generated by the grammar of interest contains all toroidal graphs. However, obviously, there are many ways of bonding $A$-labeled and $C$-labeled hyperedges in a grid as well as $B$-labeled and $D$-labeled ones resulting in many different graphs. In \cite{KreowskiKL17}, they are called \emph{pseudotori} because they include tori but also Klein bottles and other shapes of graphs.
\end{example}
In Example \ref{example_pseudotori}, the hypergraph $Z(1)$ consists of several hyperedges of type 1 attached to the same vertex. We shall use hypergraphs of such a form later so let us introduce a more general definition.
\begin{definition}
	Let $S=(S(1),\dotsc,S(k))$ be a string of labels of type $1$. Then $\&(S)$ is the hypergraph $G$ defined as follows: $V_G = \{v\}$; $E_G = \{1,\dotsc,k\}$; $att_G(i) = v$, $lab_G(i) = S(i)$ for $i \in \{1,\dotsc,k\}$.
\end{definition}
For example, the hypergraph $Z(1)$ from Example \ref{example_pseudotori} equals $\&(A,B,C,D)$.

Here are two more examples of languages generated by bonding grammars.
\begin{proposition}\label{prop_regular}
	For each fixed $k \in \Nat$, the following languages are generated by bonding grammars:
	\begin{itemize}
		\item the set of connected $k$-regular (all vertices have degree $k$) directed graphs;
		\item the set of connected directed graphs of maximum degree $\le k$.
	\end{itemize}
\end{proposition}
\begin{proof}
	Let $N = \{I,O\}$ where $I,O$ are labels of type $1$ and let $T = \{b\}$. Let $O \otimes I = b$. We define the hypergraph $Z^{(i)}_j$ for $0 \le j \le i$ as $\&(I^j, O^{i-j})$. Here $I^j$ means $I,\dotsc,I$ repeated $j$ times. Then, to generate $k$-regular graphs, take the bonding grammar $BG^{(k)}_{reg} = \left(Z^{(k)}_{reg},N,T,\otimes\right)$ where $Z^{(k)}_{reg} = \left(Z^{(k)}_0,\dotsc, Z^{(k)}_k\right)$. Indeed, if $m \cdot Z^{(k)}_{reg} \Rightarrow_\bn^\ast H$ for a connected hypergraph $H$, then, before bonding, there are exactly $k$ hyperedges attached to each vertex (labeled by either $O$ or $I$). After bonding, this results in each vertex having the degree exactly $k$; more precisely, if there are $j$ $I$-labeled hyperedges and $(k-j)$ $O$-labeled hyperedges attached to a vertex $v$ before bonding, then, after bonding, the in-degree of $v$ equals $j$ and the out-degree of $v$ equals $(k-j)$.
	
	To generate graphs of maximum degree at most $k$, take the bonding grammar $BG^{(k)}_{deg} = \left(Z^{(k)}_{deg},N,T,\otimes\right)$ where the tuple $Z^{(k)}_{deg}$ is the concatenation of the tuples $Z^{(0)}_{reg}, \dotsc, Z^{(k)}_{reg}$.
	\qed
\end{proof}
\begin{example}
	Below, the hypergraphs $Z^{(i)}_j$ are presented for $0 \le j \le i \le 2$.
	$$
	Z^{(0)}_0 = 
	\vcenter{\hbox{{\tikz[baseline=.1ex]{
					\node[vertex] (VO) at (0,0) {};
	}}}}
	\qquad
	Z^{(1)}_0 = 
	\vcenter{\hbox{{\tikz[baseline=.1ex]{
					\node[vertex] (VO) at (0,0) {};
					\node[hyperedge] (EE) at (0.6,0) {$O$};
					\draw[-] (EE) -- (VO);
	}}}}
	\qquad
	Z^{(1)}_1 = 
	\vcenter{\hbox{{\tikz[baseline=.1ex]{
					\node[vertex] (VO) at (0,0) {};
					\node[hyperedge] (EE) at (0.6,0) {$I$};
					\draw[-] (EE) -- (VO);
	}}}}
	$$
	$$
	Z^{(2)}_0 = 
	\vcenter{\hbox{{\tikz[baseline=.1ex]{
					\node[vertex] (VO) at (0,0) {};
					\node[hyperedge] (EW) at (-0.6,0) {$O$};
					\node[hyperedge] (EE) at (0.6,0) {$O$};
					\draw[-] (EW) -- (VO);
					\draw[-] (EE) -- (VO);
	}}}}
	\qquad
	Z^{(2)}_1 = 
	\vcenter{\hbox{{\tikz[baseline=.1ex]{
					\node[vertex] (VO) at (0,0) {};
					\node[hyperedge] (EW) at (-0.6,0) {$I$};
					\node[hyperedge] (EE) at (0.6,0) {$O$};
					\draw[-] (EW) -- (VO);
					\draw[-] (EE) -- (VO);
	}}}}
	\qquad
	Z^{(2)}_2 = 
	\vcenter{\hbox{{\tikz[baseline=.1ex]{
					\node[vertex] (VO) at (0,0) {};
					\node[hyperedge] (EW) at (-0.6,0) {$I$};
					\node[hyperedge] (EE) at (0.6,0) {$I$};
					\draw[-] (EW) -- (VO);
					\draw[-] (EE) -- (VO);
	}}}}
	$$
	To generate the hypergraph ${\tikz[baseline=.1ex]{
					\node[vertex] (V1) at (0,0) {};
					\node[vertex] (V2) at (0.8,0) {};
					\node[vertex] (V3) at (1.6,0) {};
					\draw[-latex, thick] (V1) -- node[above] {$b$} (V2);
					\draw[-latex, thick] (V2) -- node[above] {$b$} (V3);
	}}$ using $BG^{(2)}_{deg}$, take the hypergraphs $Z^{(1)}_0$, $Z^{(1)}_1$, and $Z^{(2)}_1$ and apply bonding as follows:
	$$
	\vcenter{\hbox{{\tikz[baseline=.1ex]{
					\node[vertex] (VO) at (0,0) {};
					\node[hyperedge] (EE) at (0.6,0) {$O$};
					\draw[-] (EE) -- (VO);
	}}}}
	\;\;
	\vcenter{\hbox{{\tikz[baseline=.1ex]{
					\node[vertex] (VO) at (0,0) {};
					\node[hyperedge] (EW) at (-0.6,0) {$I$};
					\node[hyperedge] (EE) at (0.6,0) {$O$};
					\draw[-] (EW) -- (VO);
					\draw[-] (EE) -- (VO);
	}}}}
	\;\;
	\vcenter{\hbox{{\tikz[baseline=.1ex]{
					\node[vertex] (VO) at (0,0) {};
					\node[hyperedge] (EE) at (-0.6,0) {$I$};
					\draw[-] (EE) -- (VO);
	}}}}
	\;\;\Rightarrow_\bn\;\;
	\vcenter{\hbox{{\tikz[baseline=.1ex]{
			\node[vertex] (V1) at (0,0) {};
			\node[vertex] (V2) at (1,0) {};
			\node[vertex] (V3) at (3,0) {};
			\node[hyperedge] (E1) at (1.6,0) {$O$};
			\node[hyperedge] (E2) at (2.4,0) {$I$};
			\draw[-latex, thick] (V1) -- node[above] {$b$} (V2);
			\draw[-] (E1) -- (V2);
			\draw[-] (E2) -- (V3);
	}}}}
	\;\;\Rightarrow_\bn\;\;
	\vcenter{\hbox{{\tikz[baseline=.1ex]{
					\node[vertex] (V1) at (0,0) {};
					\node[vertex] (V2) at (0.8,0) {};
					\node[vertex] (V3) at (1.6,0) {};
					\draw[-latex, thick] (V1) -- node[above] {$b$} (V2);
					\draw[-latex, thick] (V2) -- node[above] {$b$} (V3);
	}}}}
	$$
\end{example}

The language of pseudotori from Example \ref{example_pseudotori} as well as the language of connected $k$-regular graphs and the language of connected graphs of maximum degree $k$ (for $k > 3$) have unbounded treewidth. Indeed, all of them contain toroidal graphs, and it is proved in \cite{KiyomiOO16} that the toroidal graph obtained from the $n\times n$ grid by wrapping around edges between the topmost row and the bottom-most row and between the leftmost column and the rightmost column has treewidth $\ge 2n-2$. As noted in \cite{KreowskiKL17}, it is proved in \cite[Proposition 4.7]{CourcelleE12} that any language generated by a hyperedge replacement grammar has bounded treewidth, and thus none of the abovementioned languages is generated by a hyperedge replacement grammar.

\subsection{Properties of Bonding Grammars}\label{ssec_properties}
Let us discuss basic mathematical properties of bonding and of bonding grammars. We start with the following simple observations that will be used later:
\begin{enumerate}
	\item If $H \Rightarrow_\bn^\ast H^\prime$, then $H$ and $H^\prime$ have the same set of vertices.
	\item If $H \Rightarrow_\bn^\ast H^\prime$ and there is a path in $H$ between two vertices $v_1,v_2 \in V_H$, then there also is a path between them in $H^\prime$.
\end{enumerate}
The reader might have noticed that, in Definition \ref{def_BG_generates_hypergraph}, the notion of a bonding grammar generating a hypergraph $H$ was defined as the fact that $m \cdot Z \Rightarrow_{BG}^\ast H$ for some $m$, which is not the same as the corresponding definition for fusion grammars. To recall, a fusion grammar generates a hypergraph $H$ if $m \cdot Z \Rightarrow_{\fs}^\ast G$, and $H$ is a connected component of $G$. It turns out that, for bonding grammars, the two definitions are equivalent, which is thanks to the abovementioned properties of bonding.
\begin{proposition}
	Let $BG = (Z,N,T,\otimes)$ be a bonding grammar. If $m \cdot Z \Rightarrow_{BG}^\ast G$ and $H$ is a connected component of $G$, then $m^\prime \cdot Z \Rightarrow_{BG}^\ast H$ for some $m^\prime$.
\end{proposition}
\begin{proof}
	Let $Y = m \cdot Z$. Bonding does not change the set of vertices in a hypergraph, so $V_Y = V_G \supseteq V_H$. Consider the subhypergraph $Y^\prime$ of $Y$ induced by the vertices from $V_H$. We claim that $Y^\prime = m^\prime \cdot Z$. Indeed, let $Y$ consist of connected components $Y_1,\dotsc,Y_n$ (each of them equals $Z(i)$ for some $i$). If $v_1 \in V_{Y_i} \cap V_H$ and $v_2 \in V_{Y_i}$, then $v_2$ must also belong to $V_H$, because, since $v_1$ and $v_2$ are connected by a path in $Y$, they are also connected by a path in $G$, so they belong to the same connected component of $G$, i.e. to $H$. Thus, $Y^\prime$ is the disjoint union of some connected components $Y_{i_1},\dotsc,Y_{i_l}$ as desired.
	\qed
\end{proof}
A nice property of bonding is that, unlike fusion, it is reversible. We call the converse operation \emph{breaking a bond}.
\begin{definition}
	Let $H^\prime$ be a hypergraph with a hyperedge $e^\prime$ such that $lab_{H^\prime}(e^\prime) = A_1 \otimes A_2$. Let $\type(A_i) = t_i$ for $i=1,2$. Let $H$ be obtained from $H^\prime$ by removing $e^\prime$ and adding new hyperedges $e_1$, $e_2$ such that $att_H(e_1)(i) = att_{H^\prime}(e^\prime)(i)$ for $i=1,\dotsc,t_1$ and $att_H(e_2)(i) = att_{H^\prime}(e^\prime)(t_1+i)$ for $i=1,\dotsc,t_2$. Let $lab_H(e_i) = A_i$ for $i=1,2$. Then we say that $H$ is obtained from $H^\prime$ by \emph{breaking the bond} $e^\prime$.
\end{definition}
Clearly, $H$ is obtained from $H^\prime$ by breaking a bond if and only if $H^\prime$ is obtained from $H$ by bonding. Note that, for any $e^\prime \in E_{H^\prime}$, there is at most one pair of labels $A_1$ and $A_2$ such that $A_1 \otimes A_2 = lab_{H^\prime}(e^\prime)$ due to injectivity of $\otimes$. 
\begin{definition}\label{def_bond_set}
	Given a bonding grammar $BG = (Z,N,T,\otimes)$ and a hypergraph $H \in \mathcal{H}(N \cup T)$, a \emph{bond set for $H$} is a set $\Bond \subseteq E_H$ such that the hypergraph obtained from $H$ by breaking the bond $e$ for each $e \in \Bond$ is of the form $m \cdot Z$.
\end{definition}
Clearly, $H$ is generated by $BG$ if and only if there is a bond set for $H$.
\begin{remark}\label{remark_NP}
	The following non-deterministic polynomial-time algorithm enables one to check if a hypergraph $H$ can be generated by a bonding grammar $BG$:
	\begin{enumerate}
		\item Non-deterministically choose a set $\Bond \subseteq E_H$;
		\item Break the bond $e$ for each $e \in \Bond$ if possible (this is done deterministically);
		\item Check if each connected component of the resulting hypergraph is isomorphic to $Z(j)$ for some $j$. If so, then $H$ is generated by $BG$, otherwise it is not.
	\end{enumerate}
	Therefore, the membership problem for bonding grammars lies in NP. Contrast this with the membership problem for fusion grammars, which is decidable but proving this is rather hard; the known deciding algorithm is in NEXPTIME \cite{Pshenitsyn23}. In Section \ref{ssec_bonding_membership}, we shall prove that bonding grammars are able to generate NP-complete sets of hypergraphs, so the membership problem for bonding grammars is intrinsically NP-complete. For graph grammars, the NP complexity level is acceptable because many problems from the graph theory are NP-complete. In particular, hyperedge replacement grammars, also known as hypergraph context-free grammars, are also able to generate an NP-complete graph language \cite{Rozenberg97}.
\end{remark}

Example \ref{example_pseudotori} and Proposition \ref{prop_regular} show that bonding grammars can generate languages of unbounded treewidth, which hyperedge replacement grammars cannot generate. Conversely, hyperedge replacement grammars can generate languages that bonding grammars cannot generate. To show this, let us define the notions of degree and degree set.
\begin{definition}
	For a hypergraph $H$, the degree of a vertex $v \in V_H$ is the cardinality of the set $\{(e,i) \mid e \in E_H, i \in [\type(e)],\mbox{ and } att_H(e)(i) = v\}$. The \emph{degree set} of $H$ is the set of degrees of its vertices.
\end{definition}
Note that, for example, the degree of the only vertex in the hypergraph $\vcenter{\hbox{{\tikz[baseline=.1ex]{
				\node[vertex] (VO) at (0,0) {};
				\draw[latex-, thick] (VO) to[out=30,in=-30,looseness=20] node[right] {$a$} (VO);
}}}}$ equals $2$ and that the degree of the vertex $v_2$ from Example \ref{example_hypergraph} equals $3$.

\begin{proposition}
	Let $BG = (Z,N,T,\otimes)$ be a bonding grammar where $\vert Z \vert = k$ and $M(i)$ is the degree set of $Z(i)$ for $i=1,\dotsc,k$. If $BG$ generates a hypergraph $H$, then the degree set of $H$ is contained in $M(1) \cup \dotsc \cup M(k)$.
\end{proposition}
This proposition follows from the trivial observation that bonding does not change degrees of vertices. Consequently, for any bonding grammar $BG$, there is $K \in \Nat$ such that the maximum degree of each hypergraph in the language $L(BG)$ is bounded by $K$. In what follows, bonding grammars cannot generate e.g. the language of star graphs, which can be generated by a hyperedge replacement grammar \cite{Engelfriet97}. Therefore, we have proved
\begin{theorem}
	The class of languages generated by hyperedge replacement grammars neither contains nor is contained in the class of languages generated by bonding grammars.
\end{theorem}

\subsection{Bonding Grammars and Sticker Systems}\label{ssec_bonding_sticker}
Bonding grammars are designed to model sticking of DNA molecules. Nicely, it turns out that they are able to simulate regular sticker systems. Below, we present an embedding of the latter in the former (compare it with Example \ref{example_bonding_DNA}).

Let $S = (\Sigma,A,D)$ be a regular sticker system. We define the bonding grammar $\BG(S) = (Z_{S},N,T,\otimes)$ as follows. The set $N$ equals $\Sigma \cup \overline{\Sigma} \cup \{\alpha,\beta\}$ where $\overline{\Sigma} = \{\overline{a} \mid a \in \Sigma\}$ consists of new labels of the form $\overline{a}$, and $\alpha,\beta$ are also new labels. The set $T$ equals $\widetilde{\Sigma} \cup \{\varphi\}$ where $\widetilde{\Sigma} = \{\widetilde{a} \mid a \in \Sigma\}$ also consists of new labels corresponding to those from $\Sigma$. Let $\type(a) = \type(\overline{a}) = 2$ and $\type(\widetilde{a}) = 4$ for $a \in \Sigma$; let $\type(\alpha) = \type(\beta) = 1$; let $\type(\varphi) = 2$. The bond function is defined as follows: $a \otimes \overline{a} = \widetilde{a}$ (for $a \in \Sigma$); $\alpha \otimes \beta = \varphi$. Clearly, it is injective.

Now, our goal is to present the function $\tau_D:D_\Sigma \setminus \{\lambda\} \to \mathcal{H}(N \cup T)$ that transforms dominoes into hypergraphs such that sticking of two dominoes is modeled by bonding of corresponding hypergraphs. Defining this function formally would be extremely tedious, so instead we start with illustrative examples.
\begin{equation*}
\tau_D\left(
\begin{bmatrix}
	abc \\ abc \\
\end{bmatrix}
\right)
=
\vcenter{\hbox{{\tikz[baseline=.1ex]{
				\node[vertex] (U1) at (0,0.8) {};
				\node[vertex] (U2) at ($(U1)+(1.4,0)$) {};
				\node[vertex] (U3) at ($(U2)+(0.6,0)$) {};
				\node[vertex] (U4) at ($(U3)+(1.4,0)$) {};
				\node[vertex] (U5) at ($(U4)+(0.6,0)$) {};
				\node[vertex] (U6) at ($(U5)+(1.4,0)$) {};
				\foreach \i in {1,...,6}
				{
					\node[vertex] (L\i) at ($(U\i)-(0,0.7)$) {};
				}
				\node[hyperedge] (NW) at ($(U1)+(-0.6,0)$) {$\beta$};
				\draw[-] (NW) -- (U1);
				\node[hyperedge] (SW) at ($(L1)+(-0.6,0)$) {$\alpha$};
				\draw[-] (SW) -- (L1);
				\node[hyperedge] (NE) at ($(U6)+(0.6,0)$) {$\alpha$};
				\draw[-] (NE) -- (U6);
				\node[hyperedge] (SE) at ($(L6)+(0.6,0)$) {$\beta$};
				\draw[-] (SE) -- (L6);
				\draw[-latex, thick] (U2) -- node[above] {$\varphi$} (U3);
				\draw[-latex, thick] (U4) -- node[above] {$\varphi$} (U5);
				\draw[latex-, thick] (L2) -- node[above] {$\varphi$} (L3);
				\draw[latex-, thick] (L4) -- node[above] {$\varphi$} (L5);
				\node[hyperedge] (E1) at ($(U1)+(0.7,-0.35)$) {$\widetilde{a}$};
				\node[hyperedge] (E2) at ($(U3)+(0.7,-0.35)$) {$\widetilde{b}$};
				\node[hyperedge] (E3) at ($(U5)+(0.7,-0.35)$) {$\widetilde{c}$};
				\draw[-] (E1) -- node[above] {\scriptsize 1} (U1);
				\draw[-] (E1) -- node[above] {\scriptsize 2} (U2);
				\draw[-] (E1) -- node[above] {\scriptsize 4} (L1);
				\draw[-] (E1) -- node[above] {\scriptsize 3} (L2);
				\draw[-] (E2) -- node[above] {\scriptsize 1} (U3);
				\draw[-] (E2) -- node[above] {\scriptsize 2} (U4);
				\draw[-] (E2) -- node[above] {\scriptsize 4} (L3);
				\draw[-] (E2) -- node[above] {\scriptsize 3} (L4);
				\draw[-] (E3) -- node[above] {\scriptsize 1} (U5);
				\draw[-] (E3) -- node[above] {\scriptsize 2} (U6);
				\draw[-] (E3) -- node[above] {\scriptsize 4} (L5);
				\draw[-] (E3) -- node[above] {\scriptsize 3} (L6);
}}}}
\end{equation*}
$$
\tau_D\left(
	\begin{pmatrix}
		ab \\ \lambda \\
	\end{pmatrix}
	\!
	\begin{bmatrix}
		cd \\ cd \\
	\end{bmatrix}
	\!
	\begin{pmatrix}
		\lambda \\ e \\
	\end{pmatrix}	
\right)
=
\vcenter{\hbox{{\tikz[baseline=.1ex]{
				\foreach \i in {1,...,5}
				{
					\node[vertex] (U\i) at (\i*0.6,0.8) {};
				}
				\node[vertex] (U6) at ($(U5)+(1.25,0)$) {};
				\node[vertex] (U7) at ($(U6)+(0.6,0)$) {};
				\node[vertex] (U8) at ($(U7)+(1.25,0)$) {};
				\node[vertex] (L1) at ($(U5)-(0,0.7)$) {};
				\node[vertex] (L2) at ($(U6)-(0,0.7)$) {};
				\node[vertex] (L3) at ($(U7)-(0,0.7)$) {};
				\node[vertex] (L4) at ($(U8)-(0,0.7)$) {};
				\foreach \i in {5,6}
				{
					\node[vertex] (L\i) at ($(L4)+({0.6*(\i-4)},0)$) {};
				}
				\node[hyperedge] (E1) at ($(U5)+(0.625,-0.35)$) {$\widetilde{c}$};
				\node[hyperedge] (E2) at ($(U7)+(0.625,-0.35)$) {$\widetilde{d}$};
				\node[hyperedge] (NW) at ($(U1)+(-0.6,0)$) {$\beta$};
				\draw[-] (NW) -- (U1);
				\node[hyperedge] (SW) at ($(L1)+(-0.6,0)$) {$\alpha$};
				\draw[-] (SW) -- (L1);
				\node[hyperedge] (NE) at ($(U8)+(0.6,0)$) {$\alpha$};
				\draw[-] (NE) -- (U8);
				\node[hyperedge] (SE) at ($(L6)+(0.6,0)$) {$\beta$};
				\draw[-] (SE) -- (L6);
				\draw[-latex, thick] (U1) -- node[above] {$a$} (U2);
				\draw[-latex, thick] (U2) -- node[above] {$\varphi$} (U3);
				\draw[-latex, thick] (U3) -- node[above] {$b$} (U4);
				\draw[-latex, thick] (U4) -- node[above] {$\varphi$} (U5);
				\draw[-latex, thick] (U6) -- node[above] {$\varphi$} (U7);
				\draw[latex-, thick] (L2) -- node[above] {$\varphi$} (L3);
				\draw[latex-, thick] (L4) -- node[above] {$\varphi$} (L5);
				\draw[latex-, thick] (L5) -- node[above] {$\overline{e}$} (L6);
				\draw[-] (E1) -- node[above] {\scriptsize 1} (U5);
				\draw[-] (E1) -- node[above] {\scriptsize 2} (U6);
				\draw[-] (E1) -- node[above] {\scriptsize 4} (L1);
				\draw[-] (E1) -- node[above] {\scriptsize 3} (L2);
				\draw[-] (E2) -- node[above] {\scriptsize 1} (U7);
				\draw[-] (E2) -- node[above] {\scriptsize 2} (U8);
				\draw[-] (E2) -- node[above] {\scriptsize 4} (L3);
				\draw[-] (E2) -- node[above] {\scriptsize 3} (L4);
}}}}
$$
$$
\tau_D
\begin{pmatrix}
	ab \\ \lambda \\
\end{pmatrix}
=
\vcenter{\hbox{{\tikz[baseline=.1ex]{
				\foreach \i in {1,...,4}
				{
					\node[vertex] (U\i) at (\i*0.6,0.5) {};
				}
				\node[hyperedge] (NW) at ($(U1)+(-0.6,0)$) {$\beta$};
				\draw[-] (NW) -- (U1);
				\node[hyperedge] (NE) at ($(U4)+(0.6,0)$) {$\alpha$};
				\draw[-] (NE) -- (U4);
				\draw[-latex, thick] (U1) -- node[above] {$a$} (U2);
				\draw[-latex, thick] (U2) -- node[above] {$\varphi$} (U3);
				\draw[-latex, thick] (U3) -- node[above] {$b$} (U4);
}}}}
\qquad\qquad
\tau_D\begin{pmatrix}
	\lambda \\ ab \\
\end{pmatrix}
=
\vcenter{\hbox{{\tikz[baseline=.1ex]{
				\foreach \i in {1,...,4}
				{
					\node[vertex] (U\i) at (\i*0.6,0.5) {};
				}
				\node[hyperedge] (NW) at ($(U1)+(-0.6,0)$) {$\alpha$};
				\draw[-] (NW) -- (U1);
				\node[hyperedge] (NE) at ($(U4)+(0.6,0)$) {$\beta$};
				\draw[-] (NE) -- (U4);
				\draw[latex-, thick] (U1) -- node[above] {$\overline{a}$} (U2);
				\draw[latex-, thick] (U2) -- node[above] {$\varphi$} (U3);
				\draw[latex-, thick] (U3) -- node[above] {$\overline{b}$} (U4);
}}}}
$$
Generally, if 
$d = \begin{bmatrix}
	a_1 \dotsc a_n \\ a_1 \dotsc a_n \\
\end{bmatrix}$,
then $\tau_D(d)$ has a sequence of hyperedges of type 4 with labels $\widetilde{a}_1,\dotsc,\widetilde{a}_n$ consecutively connected by $\varphi$-labeled edges; there also are two $\alpha$-labeled and two $\beta$-labeled hyperedges at the corners of this hypergraph. If a domino $d = d_1d_2d_3$ with $d_2 \in D^=_\Sigma$ has the sticky end $d_1 = \begin{pmatrix}
	b_1 \dotsc b_n \\ \lambda \\
\end{pmatrix}$, then the sequence of edges with the labels $b_1,\dotsc,b_n$ consecutively connected by $\varphi$-labeled edges replaces the northwestern $\beta$-labeled hyperedge of $\tau_D(d_2)$. Other kinds of sticky ends are treated similarly.

Each part of this translation is biologically grounded. Edges labeled by symbols from $\Sigma$ and $\overline{\Sigma}$ model nucleobases; hyperedges with the labels $\alpha$ and $\beta$ model $3^\prime$ and $5^\prime$ ends of a DNA strand; hyperedges with labels from $\widetilde{\Sigma}$ model base pairs; $\varphi$-labeled edges represent phosphodiester bonds between nucleobases.

Let us also define the function $\tau_R: R_\Sigma \to \mathcal{H}(N \cup T)$. For $d \in R_\Sigma$, $\tau_R(d)$ is the hypergraph obtained from $\tau_D(d)$ by removing the northwestern $\beta$-labeled hyperedge and the southwestern $\alpha$-labeled one. Now, we are ready to define $Z_{S}$: 
\begin{definition}
	For a regular sticker system $S = (\Sigma,A,D)$, $Z_{S}$ is the tuple composed of the hypergraphs from the set $\{\tau_R(a) \mid a \in A\} \cup \{\tau_D(d) \mid (\lambda,d) \in D\}$.
\end{definition}
This completes the definition of $\BG(S)$.

\begin{theorem}
	A sticker system $S = (\Sigma,A,D)$ generates a domino $r \in R_\Sigma$ if and only if the bonding grammar $\BG(S)$ generates $\tau_R(r)$.
\end{theorem}
\begin{proof}[sketch]
To prove the ``only if'' direction it suffices to notice that, if sticking $x \cdot y$ of $x \in R_\Sigma$ and $y \in D_\Sigma$ is defined, then one can apply several bondings to $\tau_R(x)$ and $\tau_D(y)$ in such a way that the result is $\tau_R(x \cdot y)$ (see Example \ref{example_sticker_to_bonding}). Then the proof is by induction on the length of a derivation in $S$. 

Conversely, suppose that $H = \tau_R(r)$ is generated by $\BG(S)$. Assume that $r =
\begin{tabular}{|cccccc}
	\cline{1-3}
	$a_1$ & $\dotsc$ & \rEnd{$a_m$} & & & \\\cline{4-6}
	$a_1$ & $\dotsc$ & $a_m$ & $b_1$ & $\dotsc$ & \rEnd{$b_n$} \\
	\cline{1-6}
\end{tabular}
$ (other cases are treated similarly). Then
$$
H = \tau_R(r)
=\;
\vcenter{\hbox{{\tikz[baseline=.1ex]{
				\node[vertex] (U1) at (0,0.8) {};
				\node[vertex] (U2) at ($(U1)+(1.4,0)$) {};
				\node[vertex] (U3) at ($(U2)+(0.6,0)$) {};
				\node[vertex] (U4) at ($(U3)+(1,0)$) {};
				\node[vertex] (U5) at ($(U4)+(0.6,0)$) {};
				\node[vertex] (U6) at ($(U5)+(1.4,0)$) {};
				\foreach \i in {1,...,6}
				{
					\node[vertex] (L\i) at ($(U\i)-(0,0.8)$) {};
				}
				\node[vertex] (L7) at ($(L6)+(0.6,0)$) {};
				\node[vertex] (L8) at ($(L7)+(0.6,0)$) {};
				\node[vertex] (L9) at ($(L8)+(0.6,0)$) {};
				\node[vertex] (L10) at ($(L9)+(1,0)$) {};
				\node[vertex] (L11) at ($(L10)+(0.6,0)$) {};
				\node[vertex] (L12) at ($(L11)+(0.6,0)$) {};
				\node[hyperedge] (NE) at ($(U6)+(0.6,0)$) {$\alpha$};
				\draw[-] (NE) -- (U6);
				\node[hyperedge] (SE) at ($(L12)+(0.6,0)$) {$\beta$};
				\draw[-] (SE) -- (L11);
				\draw[-latex, thick] (U2) -- node[above] {$\varphi$} (U3);
				\draw[-latex, thick] (U4) -- node[above] {$\varphi$} (U5);
				\draw[latex-, thick] (L2) -- node[above] {$\varphi$} (L3);
				\draw[latex-, thick] (L4) -- node[above] {$\varphi$} (L5);
				\draw[latex-, thick] (L6) -- node[above] {$\varphi$} (L7);
				\draw[latex-, thick] (L7) -- node[above] {$\overline{b}_1$} (L8);
				\draw[latex-, thick] (L8) -- node[above] {$\varphi$} (L9);
				\draw[latex-, thick] (L10) -- node[above] {$\varphi$} (L11);
				\draw[latex-, thick] (L11) -- node[above] {$\overline{b}_n$} (L12);
				\node at ($(U3)+(0.5,0)$) {\dots};
				\node at ($(L3)+(0.5,0)$) {\dots};
				\node at ($(L9)+(0.5,0)$) {\dots};
				\node[hyperedge] (E1) at ($(U1)+(0.7,-0.4)$) {$\widetilde{a}_1$};
				\node[hyperedge] (E2) at ($(U5)+(0.7,-0.4)$) {$\widetilde{a}_m$};
				\draw[-] (E1) -- node[above] {\scriptsize 1} (U1);
				\draw[-] (E1) -- node[above] {\scriptsize 2} (U2);
				\draw[-] (E1) -- node[above] {\scriptsize 4} (L1);
				\draw[-] (E1) -- node[above] {\scriptsize 3} (L2);
				\draw[-] (E2) -- node[above] {\scriptsize 1} (U5);
				\draw[-] (E2) -- node[above] {\scriptsize 2} (U6);
				\draw[-] (E2) -- node[above] {\scriptsize 4} (L5);
				\draw[-] (E2) -- node[above] {\scriptsize 3} (L6);
}}}}
$$
Let $e^1_1,\dotsc,e^1_{m-1}$ be the upper $\varphi$-labeled edges in $H$ and let $e^2_1,\dotsc,e^2_{m+n-1}$ be the lower $\varphi$-labeled edges in $H$ (from left to right). Let also $e^0_1,\dotsc,e^0_m$ be the hyperedges of type 4 such that $e^0_i$ is labeled by $\widetilde{a}_i$. Let us denote the vertex $att_H(e^0_1)(1)$ by $\mathit{nw}$ and the vertex $att_H(e^0_1)(4)$ by $\mathit{sw}$. 

Fix some bond set $\Bond$ for $H$ and do the following with the domino representation of $r$. If $e^1_i \in \Bond$, then draw a vertical line between $a_i$ and $a_{i+1}$ in the upper strand of $r$. If $e^2_i \in \Bond$, then draw a vertical line between $a_i$ and $a_{i+1}$ in the lower strand of $r$. If $e^0_i \in \Bond$, then draw a horizontal line between $a_i$ in the upper strand and $a_i$ in the lower strand of $r$. These lines divide $r$ into several pieces $r_1,\dotsc,r_t$. Now, let us break the bond $e$ for each $e \in \Bond$; let us denote the resulting hypergraph by $F$. Clearly, breaking a bond $e^j_i \in \Bond$ in $H$ corresponds to drawing a line in $r$ according to the prodecure described above. Thus, there is a one-one correspondence between the pieces $r_1,\dotsc,r_t$ and the connected components of $F$. Let $F_i$ be the connected component of $F$ corresponding to $r_i$ (for $i=1,\dotsc,t$). Since $\Bond$ is the bond set, $F_i \in Z_S$, so either $F_i = \tau_R(d)$ or $F_i = \tau_D(d)$ for some $d$. 

Note that $e^0_1 \notin \Bond$ because if this was the case, then 
$F$ would look like
$
\vcenter{\hbox{{\tikz[baseline=.1ex]{
				\node[vertex] (U1) at (0,0.8) {};
				\node[vertex] (U2) at ($(U1)+(0.7,0)$) {};
				\foreach \i in {1,2}
				{
					\node[vertex] (L\i) at ($(U\i)-(0,0.5)$) {};
				}
				\draw[-latex, thick] (U1) -- node[above] {$a_1$} (U2);
				\draw[latex-, thick] (L1) -- node[above] {$\overline{a}_1$} (L2);
				\node at ($(U2)+(0.5,0)$) {\dots};
				\node at ($(L2)+(0.5,0)$) {\dots};
}}}}
$, so some $F_i$ would contain the $a_1$-labeled edge, to which first attachment vertex no other hyperedge is attached. This, however, cannot be the case. Indeed, $F_i$ either equals $\tau_D(d)$ or $\tau_R(d)$ for some $d$, and the definitions of $\tau_D$ and $\tau_R$ imply that either a $\varphi$-labeled or a $\beta$-labeled hyperedge must be attached to the first attachment vertex of the $a_1$-labeled edge in $F_i$. Therefore, $e^0_1$ belongs to $E_F$, so it belongs to one of its connected components, say, to $F_1$. Clearly, $F_1 = \tau_R(r_1)$ because, if $F_1 = \tau_D(r_1)$, then a $\beta$-labeled hyperedge would be attached to $\mathit{nw}$.

For each $v \in V_H \setminus \{\mathit{nw}, \mathit{sw}\}$, a hyperedge with one of the labels $\alpha,\beta$ or $\varphi$ is attached to $v$ in $H$. Breaking the bonds from $\mathcal{B}$ preserves this property because, if $e \in \mathcal{B}$ and $lab_H(e) = \varphi$, then, after breaking the bond $e$, an $\alpha$-labeled hyperedge and a $\beta$-labeled one appear instead of $e$. Thus $F_i = \tau_D(r_i)$ for $i=2,\dotsc,t$; indeed, if $F_i$ equaled $\tau_R(r_i)$, then there would exist a vertex in $V_{F_i}$ to which no hyperedge with a label $\alpha,\beta$, or $\varphi$ is attached in $F$. Concluding, $r$ is obtained from the pieces $r_1,\dotsc,r_t$ by means of sticking where the leftmost domino $r_1$ is from $A$ and, for $i=2,\dotsc,t$, $(\lambda,r_i) \in D$. Thus $r$ is generated by $S$.
\qed
\end{proof}

\begin{example}\label{example_sticker_to_bonding}
	Let 
	$
	a_0 = 
	\begin{tabular}{|cc}
		\cline{1-1}
		\lrEnd{$a$} & \\\cline{2-2}
		$a$ & \rEnd{$b$} \\
		\cline{1-2}
	\end{tabular}
	$,
	$
	d_1 = 
	\begin{tabular}{ccc}
		\cline{1-2}
		\lEnd{$b$} & \rEnd{$c$} & \\\cline{1-1}\cline{3-3}
		& \lEnd{$c$} & \rEnd{$b$} \\
		\cline{2-3}
	\end{tabular}
	$,
	$
	d_2 = 
	\begin{tabular}{cc}
		\cline{1-2}
		\lEnd{$b$} & \rEnd{$d$} \\
		\cline{1-2}
		& \\
	\end{tabular}
	$, and
	$
	r = \vcenter{\hbox{
			\begin{tabular}{|ccccc}
				\cline{1-5}
				$a$ & $b$ & $c$ & $b$ & \rEnd{$d$} \\\cline{5-5}
				$a$ & $b$ & $c$ & \rEnd{$b$} &  \\
				\cline{1-4}
			\end{tabular}
	}}
	$. Let $\Sigma = \{a,b,c,d\}$, let $A = \{a_0\}$, and let $D = \{(\lambda,d_1),(\lambda,d_2)\}$. Clearly, the sticker system $S = (\Sigma,A,D)$ generates $r$. The hypergraph $\tau_R(r)$ can be obtained from $\tau_R(a_0)$, $\tau_D(d_1)$ and $\tau_D(d_2)$ by bonding as shown below.
	
	\begin{eqnarray*}
		\vcenter
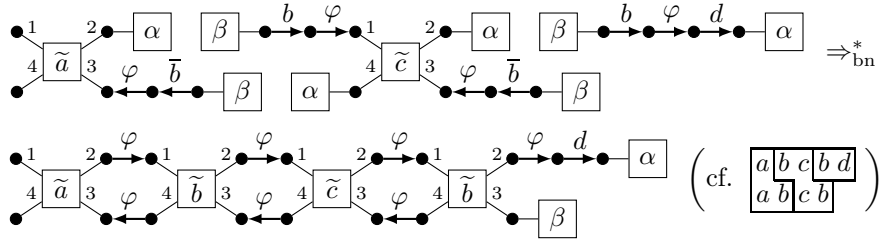
{\hbox{{\tikz[baseline=.1ex]{
						\node[vertex] (U1) at (0,0.8) {};
						\node[vertex] (U2) at ($(U1)+(1.2,0)$) {};
						\foreach \i in {1,2}
						{
							\node[vertex] (L\i) at ($(U\i)-(0,0.8)$) {};
						}
						\node[vertex] (L3) at ($(L2)+(0.6,0)$) {};
						\node[vertex] (L4) at ($(L3)+(0.6,0)$) {};
						\node[hyperedge] (NE) at ($(U2)+(0.6,0)$) {$\alpha$};
						\draw[-] (NE) -- (U2);
						\node[hyperedge] (SE) at ($(L4)+(0.6,0)$) {$\beta$};
						\draw[-] (SE) -- (L4);
						\draw[latex-, thick] (L2) -- node[above] {$\varphi$} (L3);
						\draw[latex-, thick] (L3) -- node[above] {$\overline{b}$} (L4);
						\node[hyperedge] (E1) at ($(U1)+(0.6,-0.4)$) {$\widetilde{a}$};
						\draw[-] (E1) -- node[above] {\scriptsize 1} (U1);
						\draw[-] (E1) -- node[above] {\scriptsize 2} (U2);
						\draw[-] (E1) -- node[above] {\scriptsize 4} (L1);
						\draw[-] (E1) -- node[above] {\scriptsize 3} (L2);
						\node[vertex] (2U1) at ($(U1)+(3.3,0)$) {};
						\node[vertex] (2U2) at ($(2U1)+(0.6,0)$) {};
						\node[vertex] (2U3) at ($(2U2)+(0.6,0)$) {};
						\node[vertex] (2U4) at ($(2U3)+(1.2,0)$) {};
						\node[vertex] (2L1) at ($(2U3)-(0,0.8)$) {};
						\node[vertex] (2L2) at ($(2U4)-(0,0.8)$) {};
						\node[vertex] (2L3) at ($(2L2)+(0.6,0)$) {};
						\node[vertex] (2L4) at ($(2L3)+(0.6,0)$) {};
						\node[hyperedge] (2NW) at ($(2U1)+(-0.6,0)$) {$\beta$};
						\draw[-] (2NW) -- (2U1);
						\node[hyperedge] (2SW) at ($(2L1)+(-0.6,0)$) {$\alpha$};
						\draw[-] (2SW) -- (2L1);
						\node[hyperedge] (2NE) at ($(2U4)+(0.6,0)$) {$\alpha$};
						\draw[-] (2NE) -- (2U4);
						\node[hyperedge] (2SE) at ($(2L4)+(0.6,0)$) {$\beta$};
						\draw[-] (2SE) -- (2L4);
						\draw[-latex, thick] (2U1) -- node[above] {$b$} (2U2);
						\draw[-latex, thick] (2U2) -- node[above] {$\varphi$} (2U3);
						\draw[latex-, thick] (2L2) -- node[above] {$\varphi$} (2L3);
						\draw[latex-, thick] (2L3) -- node[above] {$\overline{b}$} (2L4);
						\node[hyperedge] (2E1) at ($(2U3)+(0.6,-0.4)$) {$\widetilde{c}$};
						\draw[-] (2E1) -- node[above] {\scriptsize 1} (2U3);
						\draw[-] (2E1) -- node[above] {\scriptsize 2} (2U4);
						\draw[-] (2E1) -- node[above] {\scriptsize 4} (2L1);
						\draw[-] (2E1) -- node[above] {\scriptsize 3} (2L2);
						\foreach \i in {1,...,4}
						{
							\node[vertex] (3U\i) at ($(U1)+({7.2+0.6*\i},0)$) {};
						}
						\node[hyperedge] (3NW) at ($(3U1)+(-0.6,0)$) {$\beta$};
						\draw[-] (3NW) -- (3U1);
						\node[hyperedge] (3NE) at ($(3U4)+(0.6,0)$) {$\alpha$};
						\draw[-] (3NE) -- (3U4);
						\draw[-latex, thick] (3U1) -- node[above] {$b$} (3U2);
						\draw[-latex, thick] (3U2) -- node[above] {$\varphi$} (3U3);
						\draw[-latex, thick] (3U3) -- node[above] {$d$} (3U4);
		}}}}
		\;\;\Rightarrow_\bn^\ast
		\\
		\vcenter{\hbox{{\tikz[baseline=.1ex]{
						\node[vertex] (U1) at (0,0.8) {};
						\node[vertex] (U2) at ($(U1)+(1.2,0)$) {};
						\foreach \i in {1,2}
						{
							\node[vertex] (L\i) at ($(U\i)-(0,0.8)$) {};
						}
						\node[vertex] (L3) at ($(L2)+(0.6,0)$) {};
						\node[vertex] (L4) at ($(L3)+(1.2,0)$) {};
						\draw[latex-, thick] (L2) -- node[above] {$\varphi$} (L3);
						\node[hyperedge] (E1) at ($(U1)+(0.6,-0.4)$) {$\widetilde{a}$};
						\draw[-] (E1) -- node[above] {\scriptsize 1} (U1);
						\draw[-] (E1) -- node[above] {\scriptsize 2} (U2);
						\draw[-] (E1) -- node[above] {\scriptsize 4} (L1);
						\draw[-] (E1) -- node[above] {\scriptsize 3} (L2);
						\node[vertex] (2U1) at ($(U2)+(0.6,0)$) {};
						\node[vertex] (2U2) at ($(2U1)+(1.2,0)$) {};
						\node[vertex] (2U3) at ($(2U2)+(0.6,0)$) {};
						\node[vertex] (2U4) at ($(2U3)+(1.2,0)$) {};
						\node[vertex] (2L1) at ($(2U3)-(0,0.8)$) {};
						\node[vertex] (2L2) at ($(2U4)-(0,0.8)$) {};
						\node[vertex] (2L3) at ($(2L2)+(0.6,0)$) {};
						\node[vertex] (2L4) at ($(2L3)+(1.2,0)$) {};
						\node[hyperedge] (2SE) at ($(2L4)+(0.6,0)$) {$\beta$};
						\draw[-] (2SE) -- (2L4);
						\draw[-latex, thick] (2U2) -- node[above] {$\varphi$} (2U3);
						\draw[latex-, thick] (2L2) -- node[above] {$\varphi$} (2L3);
						\node[hyperedge] (2E1) at ($(2U3)+(0.6,-0.4)$) {$\widetilde{c}$};
						\draw[-] (2E1) -- node[above] {\scriptsize 1} (2U3);
						\draw[-] (2E1) -- node[above] {\scriptsize 2} (2U4);
						\draw[-] (2E1) -- node[above] {\scriptsize 4} (2L1);
						\draw[-] (2E1) -- node[above] {\scriptsize 3} (2L2);
						\node[vertex] (3U1) at ($(2U4)+(0.6,0)$) {};
						\node[vertex] (3U2) at ($(3U1)+(1.2,0)$) {};
						\node[vertex] (3U3) at ($(3U2)+(0.6,0)$) {};
						\node[vertex] (3U4) at ($(3U3)+(0.6,0)$) {};
						\node[hyperedge] (3NE) at ($(3U4)+(0.6,0)$) {$\alpha$};
						\draw[-] (3NE) -- (3U4);
						\draw[-latex, thick] (3U2) -- node[above] {$\varphi$} (3U3);
						\draw[-latex, thick] (3U3) -- node[above] {$d$} (3U4);
						\node[hyperedge] (E2) at ($(2U1)+(0.6,-0.4)$) {$\widetilde{b}$};
						\draw[-] (E2) -- node[above] {\scriptsize 1} (2U1);
						\draw[-] (E2) -- node[above] {\scriptsize 2} (2U2);
						\draw[-] (E2) -- node[above] {\scriptsize 4} (L3);
						\draw[-] (E2) -- node[above] {\scriptsize 3} (L4);
						\node[hyperedge] (3E1) at ($(3U1)+(0.6,-0.4)$) {$\widetilde{b}$};
						\draw[-] (3E1) -- node[above] {\scriptsize 1} (3U1);
						\draw[-] (3E1) -- node[above] {\scriptsize 2} (3U2);
						\draw[-] (3E1) -- node[above] {\scriptsize 4} (2L3);
						\draw[-] (3E1) -- node[above] {\scriptsize 3} (2L4);
						\draw[-latex, thick] (U2) -- node[above] {$\varphi$} (2U1);
						\draw[-latex, thick] (2U4) -- node[above] {$\varphi$} (3U1);
						\draw[latex-, thick] (L4) -- node[above] {$\varphi$} (2L1);
		}}}}
		\;\;
		\left(\mbox{cf. }\vcenter{\hbox{
				\begin{tabular}{|ccccc}
					\cline{1-5}
					\lrEnd{$a$} & $b$ & \rEnd{$c$} & $b$ & \rEnd{$d$} \\\cline{2-2}\cline{4-5}
					$a$ & \rEnd{$b$} & $c$ & \rEnd{$b$} &  \\
					\cline{1-4}
				\end{tabular}
		}}\right)
	\end{eqnarray*}

\end{example}

\subsection{Membership Problem for Bonding Grammars}\label{ssec_bonding_membership}

In Remark \ref{remark_NP}, we have already discussed that the membership problem for bonding grammars is in NP. Now, we are going to prove its NP-hardness.

\begin{theorem}\label{th_NP}
	Some bonding grammar generates an NP-complete language.
\end{theorem}
\begin{proof}
	In \cite{vanRooijvKNB12}, it is proved that the problem of partitioning an undirected graph of maximum degree 4 into triangles is NP-complete. The input of this problem is an undirected simple graph $G = (V,E)$ with $3q$ vertices such that the degree of each vertex is at most 4; the question is whether $V$ can be partitioned into 3-element sets $V_1$, \dots, $V_q$ such that, for each $i$, any two vertices in $V_i$ are adjacent. Let us slightly modify this problem:
	\begin{quote}
		\textbf{Problem \textsc{5Conn-PiT}}
		\\
		\textbf{Input:} an undirected simple \emph{connected} graph $(V,E)$ with $3q$ vertices such that the degree of each vertex is at most 5.
		\\
		\textbf{Question:} can $V$ be partitioned into 3-element sets $V_1$, \dots, $V_q$ such that, for each $i$, any two vertices in $V_i$ are adjacent?
	\end{quote}
	The former problem can be reduced to \textsc{5Conn-PiT} in polynomial time. Indeed, if $G = (V,E)$ consists of connected components $G_1,\dotsc,G_l$, then firstly check if each of them contains at least three vertices (otherwise, the answer to the former problem is NO). Then, for each $i = 1,\dotsc,l$, choose two vertices $v_i^1$ and $v_i^2$ in $G_i$ and add the edge $\{v_i^2,v_{i+1}^1\}$ to $G$ for each $i=1,\dotsc,l-1$. The resulting graph $G^\prime$ is connected; the degree of each its vertex is at most 5. Finally, $G^\prime$ can be partitioned into triangles if and only if so can be $G$. Indeed, if there is a partition of $V$ into 3-element sets $V_1,\dotsc,V_q$ such that any two vertices in $V_j$ are adjacent in $G^\prime$, then it is not the case that $V_j = \{v_i^2,v_{i+1}^1,v\}$ for some $i$ and $v \in V$, because this would imply that there is a path from $v_i^2$ to $v_{i+1}^1$ in $G$ going through $v$; however, $v_i^2$ and $v_{i+1}^1$ are from different connected components of $G$. 
	
	Let us now agree on how to represent an undirected simple graph as a hypergraph. Let us fix a terminal label $b$; then, a graph $G = (V,E)$ is represented by the hypergraph $\ulcorner G \urcorner = (V,E\times\{1,2\},att,lab)$ where $att(e,1) = v_1v_2$ and $att(e,2) = v_2v_1$ for $e = \{v_1,v_2\} \in E$ and where $lab(e,1) = lab(e,2) = b$. In other words, each undirected edge is represented by the following pair of hyperedges:
	$\vcenter
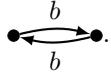
{\hbox{{\tikz[baseline=.1ex]{
					\node[vertex] (VO) at (0,0) {};
					\node[vertex] (2VO) at (1.1+0,0) {};
					\draw[-latex, thick] (VO) to[bend left = 14] node[above] {$b$} (2VO);
					\draw[-latex, thick] (2VO) to[bend left = 14] node[below] {$b$} (VO);
	}}}}$. To make drawings more compact, we are going to depict each such pair by a single double-ended arrow without a label: $\vcenter{\hbox{{\tikz[baseline=.1ex]{
	\node[vertex] (VO) at (0,0) {};
	\node[vertex] (2VO) at (1.1+0,0) {};
	\draw[latex-latex, thick] (VO) -- (2VO);
	}}}}$.
	
	Now, let us define the bonding grammar $\TBG = (\mathcal{Z},N,T,\otimes)$ where $N = \{I,O\}$ with $\type(I)=\type(O) = 1$; $T = \{b\}$; $I \otimes O = b$ (as in the proof of Proposition \ref{prop_regular}). To define $\mathcal{Z}$, we need the following definition.
	\begin{definition}
		The \emph{triangle graph} is the hypergraph $\THyp = \vcenter{\hbox{{\tikz[baseline=.1ex]{
						\node[vertex] (V2) at (0,0.8/2) {};
						\node[vertex] (V3) at (0.69/2,-0.4/2) {};
						\node[vertex] (V4) at (-0.69/2,-0.4/2) {};
						\draw[latex-latex, thick] (V2) -- (V3);
						\draw[latex-latex, thick] (V3) -- (V4);
						\draw[latex-latex, thick] (V4) -- (V2);
		}}}}$. Its vertices are denoted by $\tv_1,\tv_2,\tv_3$.
		For $k_1,k_2,k_3 \in \Nat$, the hypergraph $\THyp(k_1,k_2,k_3)$ is obtained from $\THyp$ by attaching $k_i$ new $I$-labeled hyperedges and $k_i$ new $O$-labeled hyperedges to $\tv_i$ for each $i \in \{1,2,3\}$.
	\end{definition}
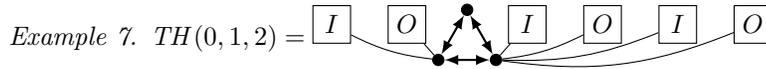
	\begin{example}\label{example_THyp}
		$
		\THyp(0,1,2) = \vcenter{\hbox{{\tikz[baseline=.1ex]{
						\node[vertex] (V2) at (0,0.8/1.8) {};
						\node[vertex] (V3) at (0.69/1.8,-0.4/1.8) {};
						\node[vertex] (V4) at (-0.69/1.8,-0.4/1.8) {};
						\draw[latex-latex, thick] (V2) -- (V3);
						\draw[latex-latex, thick] (V3) -- (V4);
						\draw[latex-latex, thick] (V4) -- (V2);
						\node[hyperedge] (I11) at (-1.8,0.25) {$I$};
						\node[hyperedge] (O11) at (-0.8,0.25) {$O$};
						\node[hyperedge] (I21) at (0.8,0.25) {$I$};
						\node[hyperedge] (O21) at (1.8,0.25) {$O$};
						\node[hyperedge] (I22) at (2.8,0.25) {$I$};
						\node[hyperedge] (O22) at (3.8,0.25) {$O$};
						\draw[-] (I11) to[bend right = 10] (V4);
						\draw[-] (O11) to[bend right = 0] (V4);
						\draw[-] (I21) to[bend left = 0] (V3);
						\draw[-] (O21) to[bend left = 10] (V3);
						\draw[-] (I22) to[bend left = 12] (V3);
						\draw[-] (O22) to[bend left = 15] (V3);			
		}}}}
		$
	\end{example}
	Finally, let $\mathcal{Z}$ be a tuple consisting of the hypergraphs $T(k_1,k_2,k_3)$ for $0 \le k_1 \le k_2 \le k_3 \le 3$ (one of these hypergraphs is shown in Example \ref{example_THyp}).
	
	We claim that, given an input $G=(V,E)$ of \textsc{5Conn-PiT}, the answer to the question of \textsc{5Conn-PiT} is YES if and only if $\ulcorner G \urcorner$ is generated by $\TBG$. Firstly, assume that the answer to the question of \textsc{5Conn-PiT} is YES, i.e. that $V$ is the disjoint union of 3-element sets $V_1,\dotsc,V_q$ such that the vertices $\{v_j^1,v_j^2,v_j^3\}$ in $V_j$ are pairwise adjacent. Let $E^\prime$ consist of the pairs $(e,1)$ and $(e,2)$ where $e = \{v_j^x,v_j^y\}$ for some $1 \le j \le q, 1 \le x < y \le 3$. Let $\Bond = E_{\ulcorner G \urcorner} \setminus E^\prime$. We claim that $\Bond$ is a bond set for $\ulcorner G \urcorner$. Indeed, after breaking all the bonds from $\Bond$, the subhypergraph of the resulting hypergraph induced by $V_j$ must be of the form $\THyp(k_1,k_2,k_3)$ where $k_i+2$ is the degree of the vertex $v^i_j$ in $G$; thus $k_i \le 3$.
	
	Secondly, assume that $\ulcorner G \urcorner$ is generated by $\TBG$, i.e. that $m \cdot \mathcal{Z} \Rightarrow_\bn^\ast \ulcorner G \urcorner$ for some $m$. The hypergraph $m \cdot \mathcal{Z}$ consists of $q$ connected components $Y_1,\dotsc,Y_q$, each of which is of the form $\THyp(k_1,k_2,k_3)$. In what follows, the vertices in $V_{Y_i}$ are pairwise adjacent in $m \cdot \mathcal{Z}$ and thus they are also adjacent in $\ulcorner G \urcorner$.
	\qed
\end{proof}

Let us conclude. Bonding grammars, defined as a modification of fusion grammars, are NP-complete and enjoy a simple algorithm for checking membership (Remark \ref{remark_NP}). Still, they generate a number of interesting languages, e.g. the language of pseudotori, the language of $k$-regular graphs, an NP-complete language. Even more importantly, bonding grammars generalise regular sticker systems in a natural and biologically adequate way, which additionally supports their motivation and even allows one to regard them as \emph{generalised sticker systems}. 

Apart from bonding, one would like to simulate other operations from the field of DNA computing by means of graph transformations. Studying extensions of bonding grammars with such operations is an open problem. It is also interesting to study how bonding grammars and fusion grammars are related.

\begin{credits}

\subsubsection{\ackname} This work was performed at the Steklov International Mathematical Center and supported by the Ministry of Science and Higher Education of the Russian Federation (agreement no. 075-15-2022-265).

\subsubsection{\discintname}
The author has no competing interests to declare that are relevant to the content of this article.

\end{credits}
%
%
%
 \bibliographystyle{splncs04}
 \bibliography{UCNC_2024}

\end{document}